%% file: __main.tex
\documentclass{article}

\input{_packages}
\input{_macros}

\newcommand{\commed}[1]{{\color{orange}#1}}

\begin{document}
\title{Bag Semantics Query Containment:\\ The CQ vs. UCQ Case and Other Stories\footnote{This research  was  supported by grant
  2022/45/B/ST6/00457 from the Polish National Science Centre (NCN).}}
\author{Jerzy Marcinkowski, Piotr Ostropolski-Nalewaja\\ {\normalsize University of Wrocław}}
\date{}

\maketitle

 \begin{abstract}
Query Containment Problem (QCP) is a
fundamental  decision problem in  query processing and optimization. 

While QCP has for a long time been completely understood for the case of set semantics,  decidability of  
QCP for conjunctive queries under multi-set
 semantics ($QCP_{\text{\tiny CQ}}^{\text{\tiny bag}}$) 
  remains 
   one of the most intriguing open problems in database theory.
Certain effort has been put, in last 30 years, to solve this 
problem and
some decidable special cases  of $QCP_{\text{\tiny CQ}}^{\text{\tiny bag}}$ were identified, as well as some
undecidable extensions,
including $QCP_{\text{\tiny UCQ}}^{\text{\tiny bag}}$.

 In this paper we introduce a new technique which produces, for a given UCQ $\Phi$, a CQ $\phi$
such that the application of $\phi$ to a database $D$ is, in some sense,
an approximation of  the application of $\Phi$ to $D$. 
Using this technique we 
could analyze the status of $QCP^{\text{\tiny bag}}$ when one 
of the queries in question is a CQ and the other is a UCQ,
and we reached conclusions which surprised us a little bit. 
We also tried to use this technique to translate the known undecidability proof for 
 $QCP_{\text{\tiny UCQ}}^{\text{\tiny bag}}$ 
 into a proof of undecidability of  $QCP_{\text{\tiny CQ}}^{\text{\tiny bag}}$. And,
 as you are going to see,
we got stopped just one infinitely small $\varepsilon$ before reaching this ultimate goal.
\end{abstract}

\input{01-introduction}

\input{02-preliminaries}

\input{03-pozytywne}

\input{03a-dowod-tw-1}

\input{04-wielomiany}

\input{05-trudne}

 \newpage

\bibliographystyle{abbrv}
\bibliography{__main}

\appendix

\input{aappendix-X}
\input{aappendix-A}
\newpage

\input{aappendix-Z}
\input{aappendix-B}
\input{aappendix-C}
\input{aappendix-D}
\input{aappendix-E}

\end{document}


%% file: _packages.tex
\usepackage{amsthm, amsmath, amssymb}
\usepackage{mathtools}
\usepackage{etoolbox}
\usepackage[capitalize]{cleveref}
\usepackage{xcolor}
\usepackage{algorithm}
\usepackage{algorithmicx}
\usepackage[noend]{algpseudocode}
\usepackage{comment}
\usepackage{enumerate}
\usepackage{multicol}
\usepackage[normalem]{ulem}
\usepackage{xspace}

\usepackage{wasysym}
\usepackage{pifont}
\usepackage{bbm}
\usepackage{scalerel} 

\usepackage{pdfpages}

\usepackage[format=plain, labelfont=it, textfont=it]{caption}

\usepackage{geometry}
\theoremstyle{plain}
\newtheorem{theorem}{Theorem}
\newtheorem{corollary}[theorem]{Corollary}
\newtheorem{lemma}[theorem]{Lemma}

\newtheorem{observation}[theorem]{Observation}

\newtheorem{Fact}[theorem]{Fact}

\newtheorem{definition}[theorem]{Definition}

  \newtheorem{innercustomthm}{Theorem}
\newenvironment{customthm}[1]
  {\renewcommand\theinnercustomthm{#1}\innercustomthm}
  {\endinnercustomthm}

\newtheorem{innercustomfac}{Fact}
\newenvironment{customfac}[1]
  {\renewcommand\theinnercustomfac{#1}\innercustomfac}
  {\endinnercustomfac}

\theoremstyle{definition}
\newtheorem{convention}[theorem]{Convention}

\newtheorem{example}[theorem]{Example}
\theoremstyle{remark}


%% file: _macros.tex
\definecolor{szary}{RGB}{225, 225, 225}

\newcommand{\function}[1]{\mathit{#1}}

\newcommand{\ddd}{\mathbbmss{d}}
\newcommand{\nnn}{\mathbbmss{n}}
\newcommand{\mmm}{\mathbbmss{m}}
\newcommand{\ccc}{\mathbbmss{c}}
\newcommand{\kkk}{\mathbbmss{k}}
\newcommand{\jjj}{\mathbbmss{j}}

\newcommand{\uuu}{\mathbbmss{u}}
\newcommand{\rrr}{\mathbbmss{r}}
\newcommand{\ttt}{\mathbbmss{t}}

\newcommand{\cccc}{  {\small \cent}}

\newcommand{\MMM}{  {\mathcal M }}

\newcommand{\dbcontained}{\leq_{\forall}}
\newcommand{\dbcontainednt}{\leq_{\forall}^{nt}}

\newcommand{\cqize}{\mathtt{cq}}
\newcommand{\ucqize}{\mathtt{ucq}}
\renewcommand{\ucqize}[1]{\raisebox{-0.2em}{\scalebox{0.2}{\includegraphics{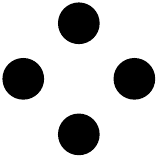}}}_{#1}}

\renewcommand{\ucqize}[1]{\mathfrak{rep}(#1)}
\renewcommand{\mu}[1]{\mathfrak{rep}_{#1}}

\newcommand{\two}{\;\text{\footnotesize\ding{203}}\;}  
\newcommand{\DDD}{{\mathbb D}}

\newcommand{\myproof}{\noindent{\sc Proof:\hspace{1mm}}}

\newcommand{\vertex}{{\mathcal V}}
\newcommand{\cansee}{\hspace{0.7mm}\text{\ding{251}}\hspace{0.7mm}}

\newcommand{\planetize}{\text{\Thorn }}

\newcommand{\trip}[2]{\mathbb T^{#2}_{#1}}

\newcommand{\bagcontained}{\leq_{\scriptsize\forall }}

\newcommand{\coef}{\mathtt{Coef}}
\newcommand{\gammas}{\ucqize{P_s}}
\newcommand{\gammab}{\ucqize{P_b}}

\newcommand{\set}[1]{\{\,#1\,\}}

\newcommand{\singleton}[1]{\{#1\}}

\newcommand{\pair}[1]{\langle\, #1 \,\rangle}

\newcounter{rscounter}
\newcounter{dbcounter}
\newcommand{\iffi}{\textit{iff} }

\newcommand{\var}[1]{\function{var}(#1)}

\newcommand{\corange}[1]{{\color{black} #1}}

\definecolor{mint}{rgb}{0, 0.8, 0.7}
\definecolor{lilac}{rgb}{0.8, 0.4, 0.7}
\definecolor{fox}{rgb}{0.9, 0.5, 0.3}

\newcommand{\piotr}[1]{{\small \textbf{P:} \color{mint}#1}}

\newcommand{\xxi}{\mathbbmss{x} }

\newcommand\wildcard[1][.60em]{%
  \makebox[#1]{%
    \kern.07em
    \vrule height.3ex
    \hrulefill
    \vrule height.3ex
    \kern.07em
  }
}

\newcommand{\nawn}[1]{\pair{#1}} 
\newcommand{\nawni}[1]{\langle #1 \rangle} 
\newcommand{\bignawn}[1]{\big\langle\, #1 \,\big\rangle} 

\newcommand{\subvenus}{\scalebox{0.75}{\venus}}

\newcommand{\Atoms}{{\footnotesize \venus}\text{-}\mathtt{Atoms}}
\newcommand{\good}{\mathtt{Good}}
\newcommand{\Planet}{\mathtt{Planet}}
\newcommand{\galaxy}{\mathtt{RClique}}

\DeclareMathOperator{\Planets}{Planets}
\newcommand{\Planetsbezv}{\Planets^{\not\subvenus}}
\newcommand{\Planetsbezvind}{\Planetsbezv}

\def\alternatesymbols{1}
\if\alternatesymbols1

\fi

\renewcommand{\cqize}{\mathfrak{cq}}
\renewcommand{\planetize}{\mathfrak{mr}}
\renewcommand{\cansee}{\twoheadrightarrow}

\makeatletter
\newcommand{\myitem}[1]{%
\item[#1]\protected@edef\@currentlabel{#1}%
}
\makeatother

%% file: 01-introduction.tex
\section{Introduction. Part 1: the general context.}\label{sec:intro}


Query Containment Problem (QCP) is one of the most  fundamental  decision problems in database query processing and optimization. 
It is formulated as follows:

\begin{center}
\begin{minipage}{0.7\linewidth}
    
\noindent
{\em The instance of QCP} are two database queries, $\Psi_s$ and $\Psi_b$.

\noindent
{\em The question} is whether 
$\Psi_s(D)\subseteq \Psi_b(D)$ holds for each database $D$.
\end{minipage}
\end{center}

In this introduction, by $\Psi(D)$  we denote
the result of applying query $\Psi$ to the database $D$.
Following \cite{MO24} we use the  subscripts $s$ and $b$  to mean
``small'' and ``big'' respectively: the QCP asks if the answer to the ``small'' query is always contained in the answer to the ``big'' one\footnote{We will also sometimes refer to $\Psi_s(D)$ as ``the s-query'' and to $\Psi_b(D)$ as ``the b-query''.}.

As most of the decision problems theoretical computer science considers,  QCP comes in many variants, depending on  parameters. In the case of QCP  usually two parameters  are considered: one is the class of queries we allow, and the second is  the
precise semantics of  $\Psi(D)$ (and -- in consequence --  the precise semantics of the symbol $\subseteq$). 
The classes of queries which have been considered in this context include 
$CQ $ (conjunctive queries), or $UCQ$ (unions of conjunctive queries) or $CQ_\neq $ (conjunctive queries 
with inequalities), or some subsets of $CQ$. The possible semantics of $\Psi(D)$ are two: either we can see $\Psi(D)$ as a relation, 
that is a {\bf set} of tuples, or as a multirelation, 
that is a {\bf multiset} also known as  a {\bf bag} of tuples. In the first case, the $\subseteq$ in the above statement
of QCP is understood to be the
set inclusion, in the second case it is the multiset inclusion.
We use natural notations to call the variants, for example $QCP^{\text{\tiny bag}}_{\text{\tiny CQ}_\neq}$ is
QCP for conjunctive queries with inequality, under bag semantics and $QCP^{\text{\tiny set}}_{\text{\tiny UCQ}}$ is $QCP$ for unions of CQs  under set semantics.

 $QCP^{\text{\tiny set}}$ has long been well understood.
It was noticed  already in 1977  that $QCP^{\text{\tiny set}}_{\text{\tiny CQ}}$ is NP-complete \cite{CM77}.
For richer query languages, $\Pi^{\text{\tiny P}}_2$-completess\footnote{
  $\Pi^\text{{\tiny P}}_2$ is the second level of
 the polynomial hierarchy, so that NP is  $\Sigma^\text{{\tiny P}}_1$ in this notation.}
was shown in  \cite{SY80} for $QCP^{\text{\tiny set}}_{\text{\tiny UCQ}}$.
 Then, in \cite{K88},  it was proven that
 $QCP^{\text{\tiny set}}_{\text{\tiny CQ},\neq,\leq}$   is also in  $\Pi^{\text{\tiny P}}_2$. Finally, in
 \cite{M97} a $\Pi^{\text{\tiny P}}_2$
 lower bound was established for this class.

But an argument can be made that
in real database systems, 
where duplicate tuples are not eliminated,
queries are usually evaluated under bag semantics, not set semantics.

Unfortunately, as it was realized in the early 1990s,
no  techniques developed for the analysis of  $QCP^{\text{\tiny set}}$ survive in the context of  $QCP^{\text{\tiny bag}}$.
In the classical paper
 \cite{CV93} the authors  observe that
the proof of the NP upper bound  for
$QCP^{\text{\tiny set}}_{\text{\tiny CQ}}$, from \cite{CM77},
 does not
survive in the bag-semantics world, and claim a $\Pi^\text{{\tiny P}}_2$ lower bound for  $QCP^{\text{\tiny bag}}_{\text{\tiny CQ}}$,
deferring the proof  to the full version of the paper.
The same observation was made also in an earlier paper \cite{CW91}, which seems to be  less well known than \cite{CV93}
Let us quote \cite{CW91} here: {\em  (...) there is almost no
theory on the properties of queries and programs that retain duplicates. The
development of such a theory is part of our future plans.}

But neither such theory was ever developed, nor did the full version of \cite{CV93} ever appear. And
no non-trivial bounds for the complexity of $QCP^{\text{\tiny bag}}_{\text{\tiny CQ}}$ were proven so far,
so not only nothing is known about the complexity of this problem but even its decidability has remained an open question for over 30 years --- despite many attempts.

Most of the effort to attack the problem was so far concentrated on the {\bf positive side}, and many results for restricted variants of $QCP^{\text{\tiny set}}_{\text{\tiny CQ}}$ were produced, which we are unable to survey here. Let us just mention that 
 decidability of $QCP^{\text{\tiny bag}}$ was shown for
 projection-free conjunctive queries  \cite{ADG10}.
This result was later extended to
the case where
$\Psi_s$ is a projection-free CQ and $\Psi_b$ is an arbitrary CQ \cite{KM19}.
The proof is by a reduction to a known decidable class of Diophantine inequalities.

Another line of attack, on the positive side, which originated from \cite{KoR11}, features the information-theoretic notion of entropy.
This technique is really beautiful, but
unfortunately the paper \cite{AKNS20}  exhibits its limitations, showing that
decidability of $QCP^{\text{\tiny bag}}_{\text{\tiny CQ}}$, even if restricted to the case where  $\Psi_b$ is an acyclic CQ,
is already  to certain
 long standing open problem in information theory.

But it is  the {\bf negative side} that is more interesting from the point of view of our paper.

\vspace{-1mm}
\section{Preliminaries. Part 1:  structures, queries and query containment.}\label{sec:prelims}

Before we continue with the introduction, it will be convenient to introduce some notations. 

We use the term  {\em structure}  to denote 
a finite relational structure over some relational schema (signature). Apart from relations we allow
 for constants (see Section \ref{o-stalych}) in the signature. We  use the letter $D$ (or $\DDD$) to denote
 structures. Notice that our input structures are {\bf not} multi-structures, 
 which means that the semantics we consider is bag-set semantics, not bag-bag semantics (see  \cite{AKNS20}). It is well known \cite{JKV06} that the computational complexity of $QCP^{\text{\tiny bag}}_{\text{\tiny CQ}}$ is the same for bag-bag semantics 
 and bag-set semantics. Also, all the results and reasonings of this paper 
 in principle remain true for bag-bag semantics (but the notational complexity increases).

 If $D$ is a structure then by $\vertex(D)$ we mean the set of vertices of $D$ (i.e. the active domain of $D$).
 For a structure $D$ and set $A\subseteq \vertex(D)$ by 
 $D\restriction_{A}$ we mean the substructure of $D$ induced by $A$, that is the structure whose set of vertices is $A$ and whose atoms are all the atoms of $D$ which only mention vertices from $A$. Following \cite{MO24}, we call a structure $D$ {\bf non-trivial} if  two constants, $\mars$ (Mars) and $\venus$ (Venus) are in the signature of $D$ and if
 $D\models\; \mars\neq \venus$.

 When we say {\em ``query''} we always mean a conjunctive query  (CQ) or a union of conjunctive queries (UCQ). We use lower case Greek letters for CQs, and upper case Greek letters for UCQs and for queries in general. 
{\bf All queries} in this paper  {\bf are Boolean}. We never explicitly
write
the existential quantifiers in front of queries
and assume all the variables to be
existentially quantified. 
If $\Phi$ is a query then $var(\Phi)$ is the set consisting of all the variables which appear in $\Phi$.
For a CQ $\phi$, by $\mathrm{canonical\_structure}(\phi)$ we mean the structure $D$ whose vertices are the variables and 
constants occurring in $\phi$ and such that for an atomic formula $A$, there is $D\models A$ if and only if $A$ is one 
of the atoms of $\phi$. If $\phi$ is a query, and $D$ is a structure, then by $Hom(\phi,D)$ we denote the set of all the 
functions $h: var(\phi)\rightarrow \vertex(D)$ which are homomorphisms from  $\mathrm{canonical\_structure}(\phi)$ to $D$. 

The {\em raison d'être} of a query is to be applied to a structure. 
For a query $\Phi$ and a structure $D$ it is standard to use the notation $\Phi(D)$ for the result of 
such application,
But our queries are complicated formulas, and our structures are usually defined by complicated expressions. We tried to 
use this standard notation but the result was unreadable, even for the authors.
 We felt we needed
an infix notation for query application, clearly separating the query from the structure. 
And we decided on $\Phi \two D$ (where $\two$ is -- obviously -- short for 
``applied two'').
If  $\phi$ is a  Boolean CQ, one would   think that 
$\phi \two D$
 can only be {\sc YES} or {\sc NO}, but since we 
consider the multiset (bag) semantics in this paper, this {\sc YES} can be repeated any  natural number of times,
depending on the number of ways $\phi$ can be satisfied in $D$. This is  formalized as: \vspace{-6mm}
$$\;\;\;\;\;\;\;\; \phi \two D \stackrel{\text{\tiny df}}{=} |Hom(\phi, D)|$$

\vspace{-1.5mm}
If  $\Phi = \bigvee_{j=1}^{\jjj}\; \phi_j$ is a UCQ, then 
$\Phi \two D  $ is defined\footnote{See \cite{HK17} for a discussion regarding the semantics of multiset union in the database theory context. } as  $ \;\Sigma_{j=1}^{\jjj} 
\nawn{\phi_j \two D}$.
We often use angle brackets, as in the last formula, to indicate that the bracketed term is a natural number. We hope this convention will slightly mitigate the pain of parsing our complicated formulas.

Recall $\Psi\two D$ is always a natural number. For a query $\Psi$ and a rational number $\rrr$ we sometimes find it convenient to write $\rrr\cdot\Psi$, to denote a new ``query'' such that $(\rrr\cdot\Psi)\two D $
equals 
$\rrr\cdot \nawn{ \Psi \two D } $. 

For queries $\Psi_s$ and $\Psi_b$ we say that query $\Psi_s$ is {\bf contained} in query $\Psi_b$
if $\Psi_s \two D \leq \Psi_b \two D $ holds for each structure $D$. We denote it as 
$\Psi_s \dbcontained \Psi_b$. We write $\Psi_s \dbcontainednt \Psi_b$ if
 $\Psi_s \two D \leq \Psi_b \two D $ holds for each non-trivial $D$. 
 $QCP^{\text{\tiny bag}}_{\text{\tiny CQ}}$ can be formulated in this language as 
 the problem whose instance are CQs $\phi_s$ and $\phi_b$ and the question is whether $\phi_s \dbcontained\phi_b$.

\vspace{-0.5mm}
  \section{Introduction. Part 2: previous works, our contribution and the structure of the paper.} 

There are,  up to our knowledge, only 3 papers so far, where negative results were shown 
for natural extensions of $QCP^{\text{\tiny bag}}_{\text{\tiny CQ}}$. 
First,  \cite{IR95}
proved that  $QCP^{\text{\tiny bag}}_{\text{\tiny UCQ}}$ is undecidable.  Then, in 2006, \cite{JKV06} proved  undecidability of $QCP^{\text{\tiny bag}}_{\text{\tiny CQ}_\neq}$. The argument here is much more complicated than the one in \cite{IR95}  and, while ``real'' conjunctive queries are mentioned in the title of \cite{JKV06},
the queries needed for the proof of this negative result require no less than $59^{10}$ inequalities.
In a recent paper \cite{MO24} it is shown that $QCP^{\text{\tiny bag}}_{\text{\tiny CQ},\neq}$ remains undecidable 
even if both $\psi_s$ and $\psi_b$ are conjunctive queries with exactly one inequality each.
It is also noticed in \cite{MO24} that the problem for  $\psi_s$ being a CQ$_\neq$ and $\psi_b$ being CQ is trivial,
and it is proven that the problem for  $\psi_s$ being a CQ and $\psi_b$ being a CQ$_\neq$ is as decidable as
$QCP^{\text{\tiny bag}}_{\text{\tiny CQ}}$ itself. Finally, \cite{MO24} show that the 
problem whose instance are two Boolean CQs, $\psi_s$ and $\psi_b$, and a natural number $c$,
and the question is whether $(c\cdot \Psi_s)\dbcontainednt \Psi_b$,
is undecidable.

All the aforementioned negative results
  (\cite{IR95}, \cite{JKV06}, \cite{MO24}) use Hilbert's 10th Problem as the source of undecidability:
  the database provides a valuation of the numerical variables, and the universal quantification from the Hilbert's Problem is 
  simulated by the universal quantification over databases. The challenge is
  how to encode the evaluation of a given polynomial using the available syntax.

  This is very simple if, like in \cite{IR95}, we deal with UCQs: a monomial in a natural way 
  can be represented as a CQ and hence a polynomial in a natural way 
  can be represented as a UCQ.
  
  In \cite{JKV06} a complicated construction was designed
  to encode an entire polynomial as one CQ. But this construction only
  works correctly for some  special databases. So this
construction is supplemented in \cite{JKV06} with a heavy anti-cheating mechanism, 
which guarantees that if $D$ is not ``special'', then
 $\psi_b \two D$ is  big enough to be greater than $\psi_s \two D$. This is done by including, in
 $\psi_b$, a sub-query which returns 1 when applied to good databases and, when applied to databases which are not good, returns
 numbers higher than anything $\psi_s$ can possibly earn thanks to ``cheating''. And
 it is this anti-cheating mechanism in  \cite{JKV06} that
  requires such huge number of inequalities. 
  
  The main idea of \cite{MO24} is a new polynomial-encoding construction, which is different from the one in 
  \cite{JKV06}, but the general philosophy is the same, and it also only works for ``correct databases''.  It is however different enough
  not to require any inequalities in the anti-cheating part, using a multiplicative constant $c $
  instead (this $c$ depends on the instance of Hilbert's 10th Problem, so it must be a part of the input and is typically huge). Then it is shown that multiplication by $c$ can be simulated by a single inequality in $\psi_b$ (while an inequality in $\psi_s$ is needed
  to enforce non-triviality).

 
In this paper we propose a new technique, significantly different than the ones from
 \cite{JKV06}, \cite{MO24},
 which we call CQ-ization. It is a generic technique, not specifically designed to 
 encode polynomials.
 
CQ-ization, for a UCQ $\Psi$, produces
 a conjunctive query $\cqize(\Psi)$ such that, for any $D$, the result of applying 
 $\cqize(\Psi)$ to $D$
 depends, in 
 a somehow predictable way, on the result  of applying $\Psi$ to $D$. 
 
 This allows us (at least to some degree)  to translate the old negative result from \cite{IR95}, for 
 $QCP^{\text{\tiny bag}}_{\text{\tiny UCQ}}$, to the realm of CQs.
 No complicated anti-cheating mechanism is needed here
 .  Using this technique we were able to easily reproduce all the results from \cite{JKV06}, \cite{MO24}
 and to prove new ones, in particular Theorem \ref{th:cq-cq} which may, more or less rightly, give the impression that the main goal in this field, that is determining  (in negative)
 the decidability status of $QCP^{\text{\tiny bag}}_{\text{\tiny CQ}}$ could  be 
 not completely out of reach.
 
The basic concepts of our technique are presented in Sections \ref{sec:running-example} and \ref{sec:cqizacja}. Then,
in Section \ref{sec:tak-samo-trudne}  the first of our two main result comes. 
  We show (and we think it is quite a surprising observation) that the possibility of having an UCQ (instead of a CQ) 
 as $\Psi_b$ does not  make the problem harder:

 \begin{theorem}\label{th:tak-samo-trudne}
      The following two claims are equivalent:\hfill --$\;\;\;QCP^{\text{\tiny bag}}_{\text{\tiny CQ}}$ is decidable.\\
      {\color{white}.}\hspace{25mm}--$\;\;\;QCP^{\text{\tiny bag}}$, restricted to instances where $\psi_s$ is a CQ and $\Psi_b$ is a UCQ, is decidable.

 \end{theorem}

\noindent
 Notice that this is not a negative result. We do not encode anything here. We just 
 CQ-ize (almost) any UCQ $\Psi_b$ and prove that nothing bad can happen.

Then we concentrate on negative results. We want (variants of) Hilbert's 10th Problem to be our source of undecidability. Therefore,
in Section \ref{sec:wielomiany},
we  explain how polynomials are represented as UCQs. This is {\bf the} natural representation,
the same as in \cite{IR95}, but using more formalized language.

In view of Theorem \ref{th:tak-samo-trudne}
one can wonder what
is the situation
if we restrict
 $QCP^{\text{\tiny bag}}$  to instances 
where 
the \corange{$s$-query $\Phi_s$} is a UCQ and 
the \corange{$b$-query $\phi_b$} is a CQ. One needs to be a bit careful here.
Imagine $D$ is the ``well of positivity", that is a database whose active domain is a 
single constant $c$, with  $D\models P(\bar c)$ for each relation $P$ in the schema 
($\bar c$ is here a tuple in which $c$ is repeated arity($P$) times). 
Then $\corange{\phi_b}(D)=1$ while
$\corange{\Phi_s}(D)=$ 
{\em the number of disjuncts in } 
$\corange{\Phi_s}$. So, obviously, containment never holds, for trivial reasons. Similar observations led the authors of  \cite{MO24} to the notion of 
{\bf non-trivial databases}, 
It follows from the main result in 
 \cite{MO24} that\footnote{This observation, however, escaped the attention of the authors of \cite{MO24}.}
if trivial counterexamples are ruled out, query containment for the case under consideration becomes undecidable:

\begin{theorem}\label{th:UCQ-duzo}
  {\em The following problem is undecidable:}\\
  Given are Boolean UCQ  $ \corange{\Phi_s} $ and Boolean CQ $\corange{\phi_b} $.
  Does $\corange{\Phi_s}(D) \dbcontainednt \corange{\phi_b}(D)$? 
\end{theorem}

In the short Section \ref{sec:drugi-dowod} we show how Theorem \ref{th:UCQ-duzo} 
can be very easily proven by  CQ-izing one of the UCQs from \cite{IR95}.
Then, in Section \ref{sec:cq-cq}
we prove our second main result:

\begin{theorem}\label{th:cq-cq}
  {\em For each rational $\varepsilon>0$ the following problem is undecidable:}\\
  Given are Boolean CQs  $ \beta_s $ and $\beta_b $.
  Does $(1+\varepsilon) \cdot \beta_s \dbcontainednt \beta_b$?
\end{theorem}

\noindent Notice that, again, the assumption that $D$ must be non-trivial is crucial,
because $(1+\varepsilon)\not\leq 1$. 

Theorem \ref{th:UCQ-duzo}, as implied by  \cite{MO24} and as proved in 
Section \ref{sec:drugi-dowod}, needs
$\corange{\Phi_s}$ to be a UCQ with potentially unbounded number of disjuncts. What would happen if we only allowed UCQs with two disjunct? 
Taking $\varepsilon=1$, as a corollary to Theorem \ref{th:cq-cq}, one gets that the problem remains undecidable:

\begin{corollary}\label{th:UCQ-dwa}
  {\em The  problem:}
  Given are Boolean CQs  $ \corange{\phi_s} $ and $\corange{\phi_b} $.
  Does $(\corange{\phi_s}\vee \corange{\phi_s})\dbcontainednt \corange{\phi_b}$?~{\em is undecidable.}
\end{corollary}
\noindent
Another, surprisingly straightforward, corollary to Theorem \ref{th:cq-cq} is 
one of the main results of \cite{MO24}:

\begin{corollary}\label{th:po-jednej-nierownosci}
  {\em The following problem is undecidable:}\\
  Given are Boolean CQs   $ \corange{\gamma_s} $ and $\corange{\gamma_b}$, each of them with at most one inequality.
  Does
  $\; \corange{\gamma_s} \dbcontained {\gamma_b}$?
\end{corollary}

\noindent
{\sc Proof of the corollary.} Assume the problem from Corollary \ref{th:po-jednej-nierownosci} is decidable. We show how to decide,
for two
 Boolean CQs  $ \beta_s $ and $\beta_b $,  if
 $2 \cdot \beta_s \dbcontainednt \beta_b$ (contradicting
 Theorem \ref{th:cq-cq} for $\epsilon=1$).
Let $P$ be a new unary relation symbol. Define:
$\alpha_s  = \; \mars\!\neq\!\venus  \wedge  P(\mars)  \wedge  P(\venus)  \wedge  P(z)  \wedge  P(z')$
and $\alpha_s  = \;  P(z)  \wedge  P(z') \wedge z\neq z'$. The following requires some focus, but otherwise is not hard to see:

\noindent
(i) for each $D$ there is $\alpha_s \two D \; \leq \; 2\cdot \alpha_b \two D$;
\hfill
(ii) there exists $D$  such that  $\alpha_s \two D \; =\; 2\cdot \alpha_b \two D$.

In order to prove (ii) just take a $D$ where the only $P$ facts are $P(\mars)$ and $P(\venus)$. To see (i) notice that $\alpha_b$ represents the number of ways a
pair of distinct elements can be drawn from a set of at least 2 elements, while $\alpha_s$ represents the number of ways any pair of elements can
be drawn from this set.
Now, let $\gamma_s=\beta_s \wedge \alpha_s   $ and $\gamma_b= \beta_b \wedge \alpha_b  $. Then  $2 \cdot \beta_s \dbcontainednt \beta_b$ if and only if
$\gamma_s  \dbcontained  \gamma_b $. \qed

\vspace{0.5mm}
 \noindent
 Could our CQ-ization technique be employed to  
 prove that $QCP^{\text{\tiny bag}}_{\text{\tiny CQ}}$ itself is undecidable (that is, to prove Theorem 
 \ref{th:cq-cq}, but for $\varepsilon=0$)? Well, first of all 
 $QCP^{\text{\tiny bag}}_{\text{\tiny CQ}}$ would need to really be undecidable. But, if it is, then 
 we are not aware of any principal reasons why CQ-ization couldn't be able to prove it. We are however 
 well aware of certain technical difficulties which stopped us one infinitely small 
 $\varepsilon$ before reaching this ultimate goal, and which we have  so far been unable to overcome. 
 

%% file: 02-preliminaries.tex
\section{Preliminaries. Part 2: more notations, some simple observations, and remarks}

We often (and sometimes silently) make use of the following obvious {but fundamental\footnote{{It may be the right moment to spot the connection
between the monomial $x^2y$ and the query $X(\wildcard) \wedge X(\wildcard) \wedge Y(\wildcard)$.}}} observation:

\vspace{-1mm}
\begin{observation}\label{obs:12}
If $\psi=\bigwedge_{k=1}^{\kkk} \psi_k$ is a \corange{CQ}, such that $\psi_k$ and $\psi_{k'}$ do not share variables for $k\neq k'$, then\\ \phantom{.}\hfill$\psi \two D \;=\; \prod_{k=1}^{\kkk}\; \nawni{\psi_k \two D}$.\hfill\phantom{.}
\end{observation}

\vspace{-1.5mm}
\subsection{Short remark about the role of constants}\label{o-stalych}

As we mentioned above, all the queries we consider in this paper are Boolean, but we allow
for constants in the language. 
As it is explained in \cite{MO24}, in the context of Query Containment, 
non-Boolean queries can always be translated to Boolean ones, for the cost of using constants, and vice versa, constants can be eliminated, but we then need non-Boolean queries. 

In this paper, we want our Theorem \ref{th:tak-samo-trudne} to be true for
general queries, also non-Boolean. But it is much easier to think about Boolean queries instead (since they return a natural number, a simple object, instead of 
a multirelations which are pain to imagine), so we need to accept presence of constants in the signature. This choice  has its downsides too: it makes Definition \ref{def:pleasant} necessary.

On the other hand, 
in order to prove 
our negative results, like 
Theorems \ref{th:UCQ-duzo} and \ref{th:cq-cq},
it is convenient to have two constants, $\venus$ and $\mars$, in the language. But again, 
using the argument from
\cite{MO24} one could trade them for free variables.


\vspace{-0.5mm}
\subsection{Some other non-standard notations}

\begin{convention}[How variables are named]\label{convention}{\em
\textbullet~ Whenever we consider a UCQ
$\Phi=\bigvee_{j=1}^{\jjj}\; \phi_j$, where each of the $\phi_j$ is a CQ, we assume 
that $var(\phi_j)$ and $ var(\phi_{j'})$ are disjoint (unless, of course, $j=j'$).

\noindent 
\textbullet~ Whenever we say that some $\Phi$ is a query over $\Sigma$ we assume that 
$x_i\not \in var(\Phi)$ for any $i\in \mathbb N$.

\noindent 
\textbullet~ 
If a variable only occurs once in a query, it does not merit a name. It is  then denoted as $\wildcard$.

\noindent 
\textbullet~ Suppose $\Phi$ is a query, $D$ is a structure, and $h$ is a function from some subset  of $var(\Phi)$ to $\vertex(D)$. Then $\Phi[h]$ is the query obtained from $\Phi$ by replacing each $v$ in the domain of $h$ with 
$h(v)$.
}
\end{convention}
For subtle technical reasons we will sometimes need to consider a slightly restricted class
of queries\footnote{It will be
wise to assume, during first reading, that there are no constants in $\Sigma$ and hence all queries are pleasant.}. This restriction, as we are  going to explain,
will not hurt the generality of our theorems:

\begin{definition}\label{def:pleasant}
A 
CQ or UCQ
is {\em pleasant} if each of its atomic formulas contains at least one 
variable.
\end{definition}

%% file: 03-pozytywne.tex
\section{CQ-ization by (a running) example}\label{sec:running-example}

In order to present our results, we will formally define, in Section \ref{sec:cqizacja}, three operations
that turn queries into other queries, and structures into other structures: \emph{relativization} ($\cansee$), \emph{CQ-ization} ($\cqize$) and \emph{marsification} ($\planetize$). 
But  let us first try to 
illustrate the main idea informally, using an example.
%




\begin{definition} Suppose $\Sigma$ is a relational signature and $V,R\not\in \Sigma$ are two 
 binary relation symbols. There may be constants in $\Sigma$, but  constants $\venus$ (Venus) and $\mars$ (Mars) are not in $\Sigma$. Denote $\Sigma^+=\Sigma \cup \{V, R, \mars, \venus\}$. 
 \end{definition}

We think of $V(a,b)$, for some vertices $a,b$ of some $D$, as an abbreviation of ``$b$ is visible from $a$''.

\vspace{2mm}
\noindent
\begin{minipage}[l]{0.47\linewidth}

\begin{example} 
Let $\Sigma = \singleton{E_1, E_2, E_3}$ and let
$\Psi = \psi_1 \vee \psi_2 \vee \psi_3$ be an UCQ, where, for each $i\in\{1,2,3\}$ there is $\psi_i = E(y_i, z_i)$.
Then \textbf{CQ-ization} of $\Psi$, denoted as $\cqize(\Psi)$, will be the CQ:
\end{example}
\vspace{-7mm}
\begin{alignat}{3}
    &\;\,\psi_1 \;\wedge\; \psi_2 \;\wedge\; &&\psi_3 \label{eq:0}  \\
    & \textstyle \bigwedge_{i = 1}^3&& R(\venus, x_i) &&\wedge R(x_i, \venus)\label{eq:1}\\
    &\textstyle \bigwedge_{i,j = 1,\, i \neq j}^3&& R(x_i,x_j)&&\wedge R(x_j, x_i)\label{eq:2}\\
    &\textstyle \bigwedge_{i=1}^3&& V(x_i,v_i)&&\wedge V(x_i, z_i) \label{eq:3}
\end{alignat}
\end{minipage}\hspace{0.6cm}
\begin{minipage}[c]{0.45\linewidth}
 \noindent
 \includegraphics[width=1\linewidth]{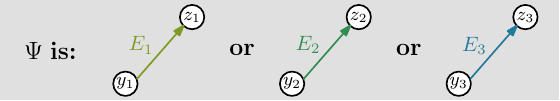}%
 \vspace{0mm}
 \noindent
 \includegraphics[width=1\linewidth]{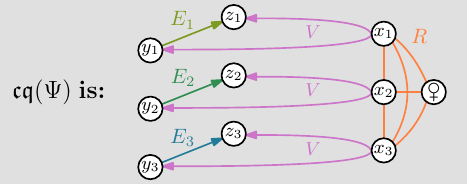}%
   \end{minipage}

\vspace{1mm}
Notice that in line (\ref{eq:0}) we now have the {\bf conjunction} of our original $\psi_1$, $\psi_2$ and $\psi_3$ (not a disjunction, like in $\Psi$).
Notice also that, compared to $\Psi$, the query $\cqize(\Psi)$ has three new variables, $x_1$, 
$x_2$ and $x_3$. We think of them as aliens\footnote{These variables have special importance in 
$\cqize(\Psi)$, and we need a name for them. Originally we called them ``guards'', but this term
could be confusing, since it is in use in the context of guarded formulas. So we decided for aliens, and found it convenient to think that the home planet of the aliens is Venus, and that they like space trips. Stay tuned.}
 sitting on some planets. The constraint from (\ref{eq:3}) says that, for each $i\in\{1,2,3\}$, query $\psi_i$ must be satisfied in the
part of the Universe visible for $x_i$. Constraints from (\ref{eq:1}) and (\ref{eq:2}) say that $x_1$,
$x_2$ and $x_3$, together with $\venus$, must form a clique in the relation $R$.

Now, as an example, imagine  the structure 
${\mathbb D}= \bigcup_{i=1}^3\singleton{E_i(a,b), E_i(b,a)}$,
Then \textbf{marsification} of ${\mathbb D}$, denoted as $\planetize({\mathbb D})$ is the structure over $\Sigma^+$ including the facts of ${\mathbb D}$ and: 

\noindent
\begin{minipage}{0.47\linewidth}
\phantom{a}
\vspace{-2.5mm}
  \begin{align}
    &\textstyle \bigwedge_{A \in \Sigma \cup \singleton{R,V}} A(\venus,\venus) \label{eq:4}\\
    &\wedge\; R(\venus,\mars), R(\mars,\venus)\label{eq:5}\\
    &\wedge\; V(\mars,a), V(\mars,b)\label{eq:6}
\end{align}  
\end{minipage}%
\hspace{0.6cm}
\begin{minipage}[c]{0.45\linewidth}
 \noindent
 \includegraphics[width=1\linewidth]{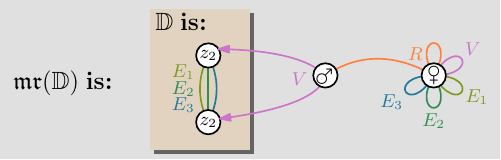}
\end{minipage}

\vspace{1mm}
Notice that the facts in lines (\ref{eq:4}) and (\ref{eq:5}) do not depend on $\mathbb D$. They will be present in
$\planetize({D})$ for any $D$. Facts in line (\ref{eq:6}) assert that every vertex of $\mathbb D$ is visible from
Mars
.

Clearly, $\Psi \two {\mathbb D} = 2+2+2=6$.
Now, let us try to calculate $\cqize(\Psi) \two \planetize({\mathbb D})$. We need to count the
homomorphisms from $\cqize(\Psi)$ to $\planetize({\mathbb D})$. The most convenient way
to count them is to group them according  to the values assigned to the aliens
$x_1$, $x_2$ and $x_3$. 

Due to the constraint from (\ref{eq:2}), which requires that the images of $x_1$, $x_2$ and $x_3$ form an
$R$-clique, there are 4 ways to assign these values (recall that  $\planetize({\mathbb D}) \models R(\venus,\venus)$).

One possibility is to send $x_1$  to $\mars$ and keep both $x_2$ and $x_3$ on $\venus$. Then,
the part of the Universe visible for $x_1$ (after it is mapped to $\mars) $is the original ${\mathbb D}$ (due to line (\ref{eq:6})), while the variables of $\psi_2$ and $\psi_3$
can only be all assigned to $\venus$. This leaves us with $\psi_1({\mathbb D})\cdot 1\cdot 1=2$ possible homomorphisms.
Second possibility is symmetric, $x_1$ and $x_3$ go to $\venus$, and each of them has only one
way to satisfy their queries in the part of the Universe they can see, and $x_2$ goes to $\mars$, from 
which she can see the original ${\mathbb D}$ and hence has
 $\psi_2({\mathbb D})=2$  ways to satisfy $\psi_2$. Third possibility is again symmetric, with
 $x_3$ being sent to $\mars$ and $x_1$ and $x_2$ staying on $\venus$.
Finally, the fourth possibility is that all three aliens, $x_1$, $x_2$ and $x_3$ stay on $\venus$.
This leads to a single homomorphism, assigning $\venus$ to all  variables in  $\planetize(\Psi)$.
%
\newcommand{\mcdot}{\hspace{-0.6mm}\cdot\hspace{-0.8mm}}
%
\begin{equation} \text{So we get:~\hspace{4mm}}
\cqize(\Psi) \two \planetize({\mathbb D})\;=\; \psi_1({\mathbb D})\mcdot 1 \mcdot 1 \;+\; 1 \mcdot \psi_2({\mathbb D}) \mcdot 1 \; +\; 1\mcdot1\mcdot\psi_3({\mathbb D})  \;+\;
1 \mcdot 1 \mcdot 1 \; =\;
\Psi({\mathbb D})+1\;\;\;\;\; \label{eq:proste}
\end{equation}
%
And (please compare!) this is exactly what Lemma
\ref{lem:plus-jeden}, at the end of Section \ref{sec:cqizacja} says. 
See how useful this Lemma can be in the context of query containment:

\begin{corollary}[of Lemma \ref{lem:plus-jeden}]\label{lem:corollary}
Suppose for two UCQs $\;\Phi_s$ and $\Phi_b$ it holds that $\Phi_s \not\bagcontained \Phi_b$, and $D$ is a counterexample for containment. Then $\cqize(\Phi_s) \not\bagcontained \cqize(\Phi_b)$, with $\planetize(D)$ as a counterexample.
\end{corollary}

\noindent
But what if $\Phi_s\bagcontained \Phi_b$? 
Does it imply that  $\cqize(\Phi_s) \bagcontained \cqize(\Phi_b)$? Or, if not, what
 additional assumptions on $\Phi_s$ and $\Phi_b$ are needed for such implication to hold?
There are many nuances here, and in order
to address these questions we need to be much more formal and precise.

\section{CQ-ization and related tricks. Let us now be formal and precise. }\label{sec:cqizacja}

 \subsection{Planets, alien R-cliques and their  space trips}

Let us start from Venus which has been assigned  a central role
in our technique. Recall the facts from (\ref{eq:4}) in Section \ref{sec:running-example}.
They were needed to make sure that every query 
will have at least one satisfying assignment in the part of the Universe visible from $\venus$. 
This is generalized in:

\begin{definition}[Good structures]\label{def:venus}\label{def:good}
 Let $\Atoms$ be the set of all atomic formulas, without variables,
that can be built using the constant $\venus$, and constants and relation symbols from
$\Sigma$, and in which $\venus$ occurs at least once. Then, by $\good$ we will mean the CQ:
$\hspace{3mm} V(\venus, \venus) \hspace{1mm} \wedge
\hspace{1mm} R(\venus, \venus) \hspace{1mm}\wedge
\bigwedge_{\alpha\in\;\Atoms} \hspace{1mm} \alpha$.\\
Structure $D$ will be called  \emph{good} if $D\models \good$.
\end{definition}

Notice that (since there are no variables in $\good$) if $D$ is a good structure, and $\Psi$ is any query, then
${\Psi\two D} = {(\Psi\wedge \good) \two D}$.
It is also easy to see that:

\begin{lemma}\label{lem:pleasant-spelnione}
    If $D$ is good, and if $\phi$ is a  pleasant CQ over
$\Sigma$ then $\phi$ can be satisfied in $D$ by the homomorphism
 mapping all the variables in $\phi$   to $\venus$.
 \end{lemma}

 Notice also  that if 
$\phi$ is not pleasant, then the  claim is not necessarily  true: suppose
$A,B\in \Sigma$ and $a$ is a constant, from $\Sigma$. Then $B(a,x)\wedge A(a) $ may not be 
satisfied in some good structure $D$. This is because formula $\Atoms$ is defined in such a way that,
while atom $B(a,\venus)$ must indeed {be present} in every good structure, there is no guarantee that $A(a) $ will be {there as well}.

A property  related to $\good$ is $\venus$-foggy. A structure is $\venus$-foggy if Venus can only 
see itself: 

\begin{definition}\label{def:foggy}
    A good structure $D$ will be called {\em $\venus$-foggy} if $\venus$ is the only $v$ satisfying  $D\models V(\venus,v)$.
\end{definition}

\noindent
Recall, that when counting the homomorphisms in Section \ref{sec:running-example}, we used the fact that an alien mapped to $\venus$ had only one way of satisfying her CQ. This was since $\planetize(D)$ is defined to always be $\venus$-foggy.


Apart of Venus we may have other planets. A vertex is a planet, if an alien can be mapped there, satisfying the constraint from (\ref{eq:1}). Not surprisingly, Mars in Section \ref{sec:running-example} was a planet:

\begin{definition}[Planets]\label{def:planets}    
By~$\Planet(x)$ we will mean  the formula $R(\venus,x)\wedge R(x,\venus)$. For a structure $D$,
over $\Sigma^+$,
 the set of elements of $\vertex(D)$ which satisfy $\Planet(x)$ will be denoted as  $\Planets(D)$ (or just $\Planets$, when $D$ is clear from the context) while
 $\Planetsbezv(D)\; = \; \Planets(D) \setminus \singleton{\venus}$.
\end{definition}

Next formula to be defined is
$\galaxy$, which formalizes and generalizes the constraint from (\ref{eq:2}) (and also includes the constraint from (\ref{eq:1})):

\begin{definition}[R-Cliques]\label{def:galaxies}   
For variables  $x_1,x_2, \ldots x_\jjj$, by $\galaxy_\jjj(x_1,x_2,\ldots x_\jjj )$ we 
mean the formula:
$$ \bigwedge_{1\leq j\leq \jjj} \Planet(x_j) \;\;\;\;\; \wedge \;\;\;\;\; \bigwedge_{1\leq j< j'\leq \jjj} R(x_j,x_{j'})\wedge R(x_j,x_{j'}).$$
\end{definition}


\subsection{Relativization. Or who can see what?}

Recall that we read $V(a,b)$ as  ``$b$ is visible from $a$''. Imagine $D$ as a universe, and $p$ as a planet in this universe. Then $seen(p,D)$ is the part of the universe
that can be seen from $p$:

\begin{definition}[Seen]
For $p\in \Planets(D)$  define
$\; seen(p,D)=D\restriction_{\{a\in \vertex(D)\;:\; D\;\models \; V(p,a)\}}$.
\end{definition}

\begin{definition}[Relativization, denoted as $\cansee$]\label{def:planetification}
   Let $\phi$ be a CQ over $\Sigma$ and let $x$ be a  variable or a constant. Then
by   $x \cansee \phi$ we denote the following CQ over  $\Sigma^+$:
$\hspace{3mm} \phi \hspace{1mm} \wedge \hspace{1mm} \Planet(x) \hspace{1mm} \wedge \bigwedge_{y\in var(\phi)} V(x,y)$.
\end{definition}

So, in order to satisfy $ x\cansee \phi$ in some $D$, the variable $x$ must be mapped onto some planet $p$, and the query
$\phi$ must be then satisfied in $seen(p,D)$.
It is  worth mentioning that the $\phi$ above may be the empty CQ. Then
 $x \cansee \phi$ is just a lonely planet gazing into darkness. Clearly:

\begin{lemma}\label{lem:widzenie}
 Let $\phi$ be a CQ over $\Sigma$,  let $D$ be a structure over $\Sigma^+$ and let $p\in \vertex(D)$ be a planet. Then:
$$ (p \cansee \phi) \two D \;\;=\;\; \phi \two seen(p,D)$$
\end{lemma}



\noindent
It is now a very easy exercise to show the following Lemma (use Lemma \ref{lem:pleasant-spelnione}):

 \begin{lemma}\label{lem:czarna-dziura}
 \hspace{-1.7mm} If $\phi$ is a  pleasant CQ over $\Sigma$,  and
  $D$ is a
 $\venus$-foggy structure over $\Sigma^+$\hspace{-0.5mm}, then ${(\venus\cansee \phi)  \two D} \hspace{-0.8mm}\;=\;\hspace{-0.8mm}1$.
\end{lemma}
%


\subsection{CQ-ization}

The most important of our operations, the one which turns UCQs into CQs, is:


\begin{definition}[CQ-ization]\label{def:cqizacja}
 Let
$\Phi=\bigvee_{j=1}^{\jjj}\; \phi_j$ be a pleasant UCQ, over $\Sigma$ (see Definition \ref{def:pleasant}), where each  $\phi_j$ is a CQ.
Then by $\cqize(\Phi)$  we mean the following CQ over $\Sigma^+$:

\vspace{-3mm}
$$\galaxy_\jjj(x_1,x_2,\ldots x_\jjj)\wedge \bigwedge_{j=1}^{\jjj} (x_j \cansee \phi_j) $$
\end{definition}
%
%

\vspace{-2mm}
\noindent
Like in Section \ref{sec:running-example}, for a UCQ $\Phi$ with  $\jjj$ disjuncts, the CQ $\cqize(\Phi)$ is produced by
creating an $\jjj$-element R-clique of new variables (recall that we see them as aliens)
and  making each of these aliens responsible for watching one of the disjuncts of $\Phi$. Notice that
a CQ is a special case of UCQ, so our CQ-ization, as defined above, can be also applied to conjunctive queries (a CQ with a
single alien variable is then produced). It is now very easy to see that:

\begin{observation}\label{obs:jeden-alien}
If $\psi$ is a CQ, and $D$ is any structure, then:
$\cqize(\psi) \two D\;=\; \textstyle\sum_{p\in \Planets} \nawn{(p\cansee \psi) \two D}$.
\end{observation}

\subsection{Counting the homomorphisms. Space trips. And two lemmas, easy but crucial.}

What we are going to do in next sections
is {\bf all} about calculating $\cqize(\Phi)\two D$ for numerous
 $\Phi$ and  $D$.

 The general method of doing it will always be the same (and the same as  in Section \ref{sec:running-example}):
 in order to count the homomorphisms from $\cqize(\Psi)$ to $D$, we will group them, count them in each group separately, and then add the results. Homomorphisms $h$ and $h'$ will fall into the same group if $h(x)=h'(x)$ for each alien variable $x$  of $\cqize(\Psi)$.
 In other words, a group will be characterized by a partial homomorphism, from the alien variables of $\cqize(\Psi)$ to $D$.
 Such partial homomorphisms will be mentioned so often that they deserve a short name, hence we call them ``trips''. This term makes sense in our narrative: trips are all the possible ways of sending aliens to the planets of $D$:

 \begin{definition}[Trips] \label{def:trip}
 For $\jjj \in\mathbb N$, and a structure $D$, a mapping $h:\{x_1,x_2,\ldots x_\jjj\}\rightarrow \Planets(D)$ will be called an $\jjj$-trip if
 $D\models \galaxy (h(x_1), h(x_2),\ldots h(x_\jjj))$. Set of all $\jjj$-trips will be denoted as $\trip{\jjj}{}(D)$, or just $\trip{\jjj}{}$
 when $D$ is clear from the context.
 \end{definition}

 If $D$ is good then $D\models R(\venus,\venus)$ and hence $h$ which maps all the arguments to $\venus$ (recall that our aliens are Venusians, so this means
 that they all stay home)
 is an $\jjj$-trip. This trip  will be called $\bar \venus$.
 Notice also that if  $p\in \Planetsbezv(D)$ and $1\leq j\leq \jjj$ then a mapping that sends
 $x_j$ to $p$ and keeps  all other aliens on $\venus$ is a $\jjj$-trip. Such trip will
be called $\bar p^j$. The set of all $\jjj$-trips of this kind, with exactly one alien mapped to a planet from $\Planetsbezv(D)$,
will be called $\trip{\jjj}{1}(D)$.
Finally, we will use the notation
$\trip{\jjj}{2\leq}(D)$ for $\trip{\jjj}{}(D)\setminus(\trip{\jjj}{1}(D)\cup\{\bar \venus\})$, that is for trips
with more than one  alien leaving $\venus$.

Using  the above defined language, there were 4 trips possible in our $\planetize({\mathbb D})$  in Section \ref{sec:running-example}, namely $\bar \mars^1$,  $\bar \mars^2$,
  $\bar \mars^3$ (all of them in $\trip{\jjj}{1}$, where of course $\jjj=3$) and $\bar \venus$.
The set $\trip{\jjj}{2\leq}$ was empty in this example.  But in general structures trips from $\trip{\jjj}{2\leq}$  exist, and they are  main
 source of trouble in
 Section \ref{sec:cq-cq}.

 To see a structure where  $\trip{\jjj}{2\leq}$ is non-empty, imagine
 ${\mathbb D}'$ being the $\planetize({\mathbb D})$ from Section \ref{sec:running-example} with the  additional facts:
$R(\venus, \saturn), R(\saturn, \venus), R(\mars, \saturn), R(\saturn, \mars)$, and possibly also with some facts of the remaining relations of $\Sigma^+$.
Then $\trip{\jjj}{1}({\mathbb D}')$ (with $\jjj=3$) will comprise 6 trips, namely $\bar \mars^j$ and $\bar \saturn^j$ for each
$j\in\{1,2,3\}$. And $\trip{\jjj}{2\leq}({\mathbb D}')$ will comprise 6 trips too, because (in this particular case)
trips from $\trip{\jjj}{2\leq}({\mathbb D}')$ will be bijections between the three aliens and the three planets {Venus, Mars, and Saturn} (notice that there would be even more
trips in $\trip{\jjj}{2\leq}$ if we also added  $R(\saturn, \saturn)$ to the structure).

Now let $\Phi$ be like in Definition \ref{def:cqizacja} and let $h\in \trip{\jjj}{}(D) $ for some structure $D$ over $\Sigma^+$.  Recall that
notation  $\cqize(\Phi)[h]$  was defined by Convention \ref{convention}. Then $\cqize(\Phi)[h] \two D$ is exactly the number of homomorphisms,
from $\cqize(\Phi)$ to $D$, which agree with $h$ on the alien variables. The following lemma  is now quite obvious (notice that
the sum in Lemma \ref{lem:pi-rachunek-sredni} reflects the addition from (\ref{eq:proste}) in Section \ref{sec:running-example}):

\begin{lemma}\label{lem:pi-rachunek-sredni}
Let $\Phi$ and  $D$  be as above. Then  \;
$\cqize(\Phi) \two D \;=\;  \sum_{h\in \trip{\jjj}{}(D)}\;  \nawni{\cqize(\Phi)[h] \two D}$ 
\end{lemma}

Next lemma, equally crucial and equally obvious, and also already silently used in Section \ref{sec:running-example}, tells us how to count the homomorphisms in each group
(notice that the product in Lemma \ref{lem:pi-rachunek-maly} reflects the multiplications from \ref{eq:proste}):

\begin{lemma}\label{lem:pi-rachunek-maly}
Let $\Phi$ be as above.
Then:
 $\;\;\cqize(\Phi)[h] \two D\;\;=\;\; \prod_{j=1}^{\jjj}\;\;\nawni{\phi_j \two seen(h(x_j),D)}$
\end{lemma}

\noindent
{\sc Proof:}
By Observation \ref{obs:12} we get: $\cqize(\Phi)[h] \two D=\prod_{j=1}^{\jjj} \nawni{ h(x_j) \cansee \phi_j \two D}$. Now use Lemma \ref{lem:widzenie}.
\qed




\subsection{Marsification}

Recall that $\mars$ (Mars) is another special constant we use.
Marsification is an operation which takes, as its input, a structure $D$ over $\Sigma$ and returns a new structure $\planetize(D)$, over
$\Sigma^+$, defined as:
\vspace{-0.5mm}
$$ \mathrm{canonical\_structure}(\good\wedge \Planet(\mars)) \;\cup\; D \;\cup\; \{V(\mars,a): a\in \vertex(D)\} $$ 
\vspace{-0.25mm}
Notice, that this formalizes what we did in Section \ref{sec:running-example}: $\mathrm{canonical\_structure}(\good\wedge \Planet(\mars))$ reflects (\ref{eq:4})
and (\ref{eq:5})  while $\{V(\mars,a): a\in \vertex(D)\}$ reflects
(\ref{eq:6}).

\begin{lemma}\label{lem:plus-jeden}
Let  $\Phi=\bigvee_{j=1}^{\jjj}\; \phi_j$ be a pleasant UCQ and let $D$ be a structure, both  over $\Sigma$. Then
$$ \cqize(\Phi) \two \planetize(D)\;\;=\;\; 1\,+\, \Phi \two D.$$

\end{lemma}

The idea of the proof of Lemma \ref{lem:plus-jeden} {should be} clear for Readers who read Section \ref{sec:running-example},
so we defer it to Appendix \ref{app:X}. Notice, however, that there must be some unexpected nuances lurking there, if we needed to assume that
$\Phi$ is pleasant.

%% file: 03a-dowod-tw-1.tex

\section{Proof of Theorem \ref{th:tak-samo-trudne} }\label{sec:tak-samo-trudne}

Now we are ready to prove Theorem \ref{th:tak-samo-trudne}. It postulates that the following two claims are equivalent:
  \begin{align}
      &\text{$QCP^{\text{\tiny bag}}$, restricted to instances where $\psi_s$ is a CQ and $\Psi_b$ is a UCQ, is decidable.}\label{eq:7h}\\
      &\text{$QCP^{\text{\tiny bag}}_{\text{\tiny CQ}}$ is decidable.}\label{eq:7c}
  \end{align}

\noindent
Clearly,  $(\ref{eq:7h})$ implies
 $(\ref{eq:7c})$. 
In order to show  the opposite implication, we introduce one more claim:
\begin{equation}
   \text{$QCP^{\text{\tiny bag}}$, restricted to instances where $\psi_s$ is a 
 CQ and $\Psi_b$ is a pleasant UCQ,
is decidable.}\label{eq:7s} 
\end{equation}
Now, to prove Theorem \ref{th:tak-samo-trudne}, we show 
that  $(\ref{eq:7s}) \Rightarrow (\ref{eq:7h})$ and that
$(\ref{eq:7c}) \Rightarrow (\ref{eq:7s})$.
Proof of  $(\ref{eq:7s})\Rightarrow (\ref{eq:7h})$ is by a standard
reasoning which does not use the idea of CQ-ization and we defer it to \cref{app:A}.\vspace{1mm}\\
{\bf The rest of this section is devoted to the proof of the 
$(\ref{eq:7c})\Rightarrow (\ref{eq:7s})$ implication.}\vspace{1mm}\\
~
Suppose we are given a CQ  $\psi_s$ and a  pleasant UCQ\footnote{
The number of disjuncts in the $b$-query will always in this paper be denoted as $\mmm$. The number 
of disjuncts in the $s$-query (if the $s$-query is a UCQ) will be always $\kkk$.
}
$\Psi_b=\bigvee_{m=1}^{\mmm}\; \phi_m$
(both over $\Sigma$) and we want to decide whether $\psi_s\bagcontained\Psi_b$.
Let $\eta_0$  be the formula $\bigwedge_{ m=1}^{\mmm}\;V(\venus,\wildcard)$. Notice that
if $D$ is  $\venus$-foggy  then $V(\venus,\wildcard)\two D=1$ and hence  $\eta_0\two D=1$. And if $D$ is good but not
$\venus$-foggy then $\eta_0\two D\geq 2^\mmm\geq \mmm$.\vspace{1mm}

\noindent
Define CQs $\gamma_s$ and $\gamma_b$ over $\Sigma^+$ as:
$\;\;\;\; \gamma_s\; = \;\good \wedge \cqize(\psi_s)\;\;\;\ \text{and}\;\;\;\; \gamma_b\;=\;
\;\eta_0 \wedge \cqize(\Psi_b)$.
\begin{equation}
\hspace{-3mm}
\text{Now,
all we  need to show is that:} \hspace{23mm}
\psi_s \bagcontained \Psi_b \quad\iff\quad
\gamma_s \bagcontained \gamma_b \hspace{18mm} \label{ding80}
\end{equation}

\noindent
Let us explain what is going on here. We assume $(\ref{eq:7c})$. We want to know whether $\psi_s\bagcontained\Psi_b$. So we
{\em almost} replace  $\psi_s$ and $\Psi_b$ with $\cqize(\psi_s)$ and $\cqize(\Psi_b)$ and feed them to the (hypothetical) algorithm for $QCP^{\text{\tiny bag}}_{\text{\tiny CQ}}$. ``{Almost}'', because
our $\gamma_s$ is equivalent to $\cqize(\psi_s)$  on good structures, {but} on non-good ones it simply returns 0. And because, our $\gamma_b$
is equivalent to $\cqize(\Psi_b)$ only on $\venus$-foggy structures, while on non-$\venus$-foggy ones $\gamma_b$ is given a huge bonus which makes
it much easier for it  to win against $\gamma_s$.\smallskip\\
\noindent
{\bf Proof of (\ref{ding80}) $(\neg \Rightarrow \neg)$.}
Recall that, for each structure $D$, the structure $\planetize(D)$ is good and \venus-foggy.
Hence, from  \cref{lem:corollary}, if
$\psi_s\not\bagcontained \Psi_b$ then
$\gamma_s\not\bagcontained \gamma_b $.\smallskip\\
\noindent
{\bf Proof of (\ref{ding80}) $( \Rightarrow )$.}
Assume that $\psi_s \bagcontained \Psi_b$ and take a structure $\DDD$ over $\Sigma^{+}$
(which will now be fixed).
Of course, we can assume that $\DDD$ is good; otherwise $\gamma_s {\two} \DDD=0$ and there is nothing
to prove. \vspace{1.2mm}\\
What remains to be shown is that:
\vspace{-6.4mm}
\begin{equation} \hspace{60mm} \cqize(\psi_s)\two \DDD\;\; \leq \;\; \gamma_b \two \DDD  \hspace{10mm} \label{eq:psik}\end{equation}
%
%

\noindent
In order to prove (\ref{eq:psik}) first of all recall that we know from Observation \ref{obs:jeden-alien} that:
 $$\cqize(\psi_s) \two \DDD \;\;=\;\; \textstyle\sum_{p\in \Planets} \nawn{(p\cansee \psi_s) \two \DDD} \;\;=\;\;
   (\venus\cansee \psi_s) \two \DDD \;\;+\;\; \textstyle\sum_{p\in \Planetsbezv} \nawn{(p\cansee \psi_s) \two \DDD}$$

   \vspace{1mm}
   \noindent
It is also easy to see that:\vspace{-0.59cm}
\begin{alignat*}{3}
   \phantom{a}\hspace{3cm} \gamma_b \two \DDD \;&=\; &&\textstyle\sum_{t \in\trip{\mmm}{}}  \,\gamma_b[t] \two \DDD
\;\;\geq\\
&\geq \;\; \gamma_b[\bar \venus] \two \DDD\;+\;&&\textstyle\sum_{t \in \trip{\mmm}{1}}
\;\gamma_b[t] \two \DDD \;\;\geq\vspace{1mm}\\
&\geq\;\; \gamma_b[\bar \venus] \two \DDD\;+\;&&\textstyle\sum_{t \in \trip{\mmm}{1}} \; \nawni{\cqize(\Psi_b)[t] \two \DDD}
\;\;\;\;
\end{alignat*}

  \vspace{1mm}
\noindent
The equality above follows from Lemma \ref{lem:pi-rachunek-sredni}. First inequality follows from the fact that
$\trip{\mmm}{1} \cup \singleton{\bar \venus} \subseteq \trip{\mmm}{}$ and second inequality follows directly from the definition of $\gamma_b$.
%

\noindent
In consequence, in order to prove (\ref{eq:psik}) it will  be enough to show the following two lemmas:

\begin{lemma}\label{lem:niewiemco-planety}
$\textstyle\sum_{p\in \Planetsbezv} \nawn{(p\cansee \psi_s) \two \DDD} \;\leq\; \textstyle\sum_{t \in \trip{\mmm}{1}} \; \nawni{\cqize(\Psi_b)[t] \two \DDD}$.
\end{lemma}
\begin{lemma}\label{lem:niewiemco-wenus}
    $(\venus\cansee \psi_s) \two \DDD \;\leq\;\gamma_b[\bar \venus] \two \DDD$.
\end{lemma}

\noindent
Informally, Lemma \ref{lem:niewiemco-planety} says that the number of such homomorphisms from $\cqize(\psi_s)$ to $\DDD$, that map the only alien in
 $\cqize(\psi_s)$ to some non-$\venus$ planet,  is already dominated by the number of homomorphisms from $\cqize(\Psi_b)$ to $\DDD$, which move exactly one
 of its aliens to some non-$\venus$ planet (while the remaining ones sit on $\venus$). No ``bonus'' {from $\eta_0$} is needed here.

 On the other hand, Lemma \ref{lem:niewiemco-wenus} deals with the number of homomorphisms from $\cqize(\psi_s)$ to $\DDD$, which map the only alien in
 $\cqize(\psi_s)$ to $\venus$. One would think that this number should be dominated by the number of homomorphisms from $\cqize(\Psi_b)$ to $\DDD$, which keep
 all the aliens  on $\venus$. But, as you are going to see in the {\sc proof of Lemma \ref{lem:niewiemco-wenus}} (which is deferred to {\sc Appendix} \ref{app:Z}), this is only true for $\venus$-foggy structures $\DDD$. Otherwise we need the ``bonus'' offered by query $\eta_0$.

\vspace{2mm}
\noindent {\sc Proof of Lemma \ref{lem:niewiemco-planety}.}
Recall that $\bar p^m$ (where  $p\in \Planetsbezv$ and $1\leq m \leq \mmm$) is a trip that maps $x_m$ to $p$ and all the other aliens to $\venus$.
In other words, $\bar p^m(x_i)=p$ if $i=m$ and $\bar p^m(x_i)=\venus$ if $i\neq m$. Clearly:
\begin{equation} \cqize(\Psi_b)[\bar p^m] \two \DDD \quad =\quad
    \textstyle\prod_{i=1}^{\mmm} \;\;\nawn{\phi_i \two seen(\bar p^m(x_i),\;\DDD)}
    \quad\geq\quad \phi_m \two seen(p,\;\DDD) \label{eq:gwiazdka} \end{equation}
    The above equality follows from Lemma \ref{lem:pi-rachunek-maly}, and the inequality holds since $\DDD$ is good. So:

\vspace{-3mm}
\begin{alignat*}{4}
    \textstyle\sum_{t \in \trip{\mmm}{1}}  \nawn{\cqize(\Psi_b)[t] \two \DDD}
&\;\;=\;\; \textstyle\sum_{p \in \Planetsbezv} \textstyle\sum_{m=1}^{\mmm}&&\nawn{\cqize(&&\Psi_b)[\bar p^m] \;\two\; \DDD&&}\tag{Definition of $\trip{\mmm}{1}$}\\
&\;\;\geq\;\;  \textstyle\sum_{p \in \Planetsbezv} \sum_{m=1}^{\mmm}\; &&\nawn{\phi_m &&\two seen(p,\DDD)&&}\tag{From (\ref{eq:gwiazdka})}\\
&\;\;=\;\;\textstyle\sum_{p \in \Planetsbezv} &&\nawn{\Psi_b &&\two seen(p,\DDD)&&}\tag{Definition of $\Psi_b$}\\
&\;\;\geq\;\; \textstyle\sum_{p\in \Planetsbezvind} &&\nawn{\psi_s &&\two seen(p,\DDD)&&}\tag{Since $\psi_s\bagcontained \Psi_b$}\\
&\;\;=\;\;\textstyle\sum_{p\in \Planetsbezvind} &&\nawn{(p&&\cansee \psi_s) \two \DDD&&} \tag*{\qed} 
\end{alignat*}

%% file: 04-wielomiany.tex
\section{Polynomials, and how to represent them}\label{sec:wielomiany}

Now we will show how the machinery developed in Section \ref{sec:cqizacja} can be used to prove
negative results, like 
Theorems \ref{th:UCQ-duzo} and \ref{th:cq-cq}. Our proofs use undecidability of Hilbert's 10th Problem, so we need a language to talk about polynomials.

Whenever we say ``monomial'' in this paper 
we mean a monomial, with coefficient $1$,
over numerical\footnote{We call variables ranging over {$\mathbb N$} {\em numerical variables} to
distinguish them from the first order logic variables in the queries. } 
variables $\xxi_1, \xxi_2,\dots  \xxi_\nnn$. 
The set of all possible monomials will be denoted as $\MMM$.
Whenever in this paper we say ``polynomial'' we mean a sum of such monomials. 
 This causes no loss of generality since
 (contrary to what our algebra textbook says),
 we do not assume that the monomials in one polynomial are pairwise distinct. 
Sometimes, however, we need to 
see the textbook polynomials, with natural coefficients, behind our polynomials: 

\begin{definition}
For a polynomial $P=\sum_{j=1}^{\jjj}\;M_j$ and  $M\in \MMM$ define:
$\;\coef(M,P)= |\{j : M_j=M \}| $.
\end{definition}

\noindent
 Whenever we say ``valuation'' we  mean a function $\Xi:\{\xxi_1, \xxi_2,\ldots  \xxi_\nnn\}\rightarrow \mathbb N$.
 For a monomial $M$ (or a polynomial $P$) and a valuation $\Xi$, by $M(\Xi)$ (or $P(\Xi)$) we 
 mean the result of applying $M$ to $\Xi(\bar \xxi)$.

Let us first explain how monomials can be  easily represented as CQs.  We will
 follow
\cite{IR95} here, but using different language. Let $\Sigma=\{X_1, X_2, \ldots X_\nnn\}$ be a signature of unary relation symbols.
Notice that there are no constants in $\Sigma$ and thus all queries we consider are pleasant. 

\begin{definition}\label{def:straszne-emy-dwa}
Let $M = \xxi_{a_1}\xxi_{a_2}\ldots\xxi_{a_\ddd}$ be a monomial of degree $\ddd$. Then define $\mu{M}$ as $\bigwedge_{i = 1}^{\ddd} X_{a_i}(\wildcard)$.
\end{definition}

\noindent
So, for example, if $M$ is $\xxi_2\xxi_4\xxi_2$ 
{then} $\mu{M}$ is $X_2(\wildcard)\wedge X_4(\wildcard)\wedge X_2(\wildcard)$. 
If $M=1$ is the (unique) monomial of degree 0
then $\mu{M}$ is the empty CQ. 

For a structure $D$, the cardinalities of the relations $X_n$ in $D$ define a valuation, in the natural way:

\begin{definition}\label{def:straszne-emy}
For a structure $D$, define $\;\Xi_D:\{\xxi_1, \xxi_2,\ldots  \xxi_\nnn\}\rightarrow \mathbb N$
as $\;\Xi_D(\xxi_n)= \nawni{X_n(\wildcard) \two D}$.
For~a~valuation
$\Xi:\{\xxi_1, \xxi_2,\ldots  \xxi_\nnn\}\rightarrow \mathbb N$
let $D_\Xi$ be some structure, over $\Sigma$, which has, for each $n$, exactly $\Xi(\xxi_n)$
vertices satisfying $X_n$.
\end{definition}

\noindent
Obviously, $\Xi_{D_\Xi}=\Xi$.  It is also not hard to see that $D$ indeed does to $\mu{M}$ what $\Xi_D$ does to $M$:
\begin{lemma}\label{lem:na-mon-ok-global}
 Let $D$ be a structure and let $M$ be a monomial. Then  $\;M(\Xi_D)=\mu{M} \two D$.

\end{lemma}

\noindent
Clearly, if monomials are  represented as CQs, then polynomials are  represented as UCQs:

\begin{definition}[Representing polynomials]\label{def:poly-to-cq}
\mbox{For a polynomial $P=\sum_{\hspace{-0.4mm}j=1}^{\jjj} \hspace{-0.4mm}M_j$ let: $\ucqize{P}=\bigvee_{\hspace{-1mm}j=1}^{\jjj}\mu{M_j}$}
\end{definition}

\begin{lemma}\label{lem:poly-reprezentacja}
 Let $D$ be a structure, and let $P$ be a polynomial. Then  $\;P(\Xi_D)= \ucqize{P} \two D$.
\end{lemma}

\noindent
As   corollary  we get that, for polynomials $P_s$ and $P_b$, inequality
$\ P_s(\Xi) \leq P_b(\Xi)$ holds for every  valuation $\Xi$ if and only if
$\ucqize{P_s}\dbcontained \ucqize{P_b}$. In view of undecidability of Hilbert's 10th Problem it
immediately implies the main result from \cite{IR95}, that is undecidability of $QCP^{\text{\tiny bag}}_{\text{\tiny UCQ}}$.

Now, our  plan is
to CQ-ize $\ucqize{P_s}$ and/or $\ucqize{P_b}$ to get undecidability  for more restricted
fragments.
As the first step on this path we  ``relativize'' the notion of valuation: each planet will see its own valuation. And hence, each planet will have its own values for monomials and polynomials:

\begin{definition}\label{def:lokalne-wartosciowanie}
For a structure $D$ and for $p\in \Planets(D)$
define valuation $\Xi^p$ 
as $\Xi_{seen(p,D)}$.
\end{definition}

\noindent
So, for example,
$\Xi^{\tiny \mercury}(\xxi_3)$
is the number of vertices satisfying
$X_3$ which are visible from Mercury ($\small \mercury$).
The next lemma is an obvious modification of Lemmas \ref{lem:na-mon-ok-global} and \ref{lem:poly-reprezentacja}:

\begin{lemma}\label{lem:na-mon-ok-local} Let $M$ be a monomial and $P=\sum_{j=1}^{\jjj}\;M_j$ a polynomial. Let $D$ be a structure and $p\in \Planets(D)$. Then:
 \hfill  $\;M(\Xi^p) =  (p \cansee \mu{M} ) \two D\;\;\;\;$ and
 $\;\;\; \;P(\Xi^p) = (\bigvee_{k=1}^{\kkk} (p \cansee \mu{M_k})) \two D$.
\end{lemma}


\section{Proof of Theorem \ref{th:UCQ-duzo}}\label{sec:drugi-dowod}

In this very short section, we will prove Theorem  \ref{th:UCQ-duzo}.
The goal is to showcase our technique, and also to give the Reader an opportunity to get used to our language.
Our source of undecidability will be the following version of Hilbert's 10th Problem:

\begin{Fact}\label{source-of-und}
{\em The following problem is undecidable:}\\
Given are polynomials $P_s$ and $P_b$. Does $\ P_s(\Xi) \leq   1+P_b(\Xi)$ hold for every  valuation?
\end{Fact}
From now on, till the end of this section,  we assume that $ P_s=\sum_{k=1}^{\kkk}\, M_k^s$ and 
$P_b=\sum_{m=1}^{\mmm}\, M_m^b$
 are two fixed 
polynomials.
Now, in order to prove Theorem \ref{th:UCQ-duzo}, we need to construct a
UCQ $\Phi_s$, and a CQ $\phi_b$, such that the two conditions are equivalent:\\
\noindent
\begin{minipage}{0.59\textwidth}
\begin{equation}
             P_s(\Xi) \leq 1+P_b(\Xi) \text{\; holds for every  valuation \;} \Xi\;\; \label{9clubsuit}
\end{equation}    
\end{minipage}
\begin{minipage}{0.12\textwidth}{\color{white} nicin }
\end{minipage}
\begin{minipage}{0.26\textwidth}
    \begin{equation}
             \text{$\Phi_s \dbcontainednt \phi_b $}\label{9spadesuit}
         \end{equation}
\end{minipage}


\vspace{1.5mm}
 Recall that in the context of
 Theorem \ref{th:UCQ-duzo} we only consider non-trivial structures, which satisfy $\mars\neq \venus$.  We are now ready to define our $\Phi_s$ and $\phi_b$. Let
   \[ \Phi_s \;=\; \good \,\wedge\, \Planet(\mars) \,\wedge\,\textstyle \bigvee_{k=1}^{\kkk} (\mars \cansee \mu{M_k^s}) \quad\quad\text{and}\quad\quad\;\; \phi_b\;=\;\cqize( \ucqize{P_b}) \]
  Since conjunction distributes over disjunction, the $\Phi_s$ as defined above is a UCQ having 
  the query  $\good \wedge \Planet(\mars)$ as a subquery in each of its disjuncts. This subquery has no variables, so (as long as it is satisfied)
 it does not affect the results of applying the disjuncts to structures. And thus (recall Lemma \ref{lem:na-mon-ok-local}) $\Phi_s$, applied to some $D$, returns
 the value of the polynomial $P_s$ on Mars.

Once  $\Phi_s$ and $\phi_b$ are defined, our last step is to show that
 $(\ref{9clubsuit}) \Leftrightarrow (\ref{9spadesuit}) $.
To see that $\;\neg(\ref{9clubsuit}) \Rightarrow \neg(\ref{9spadesuit}) $ assume that there exists
valuation $\Xi$
such that
$\ P_s(\Xi) > 1+P_b(\Xi)$.
Take
$\DDD= \planetize(D_\Xi) $ and recall that $\planetize(D_\Xi)$ is defined as a structure with two planets ($\venus$ and $\mars$) such that the
valuation seen from Mars is $\Xi$. So:

\vspace{-4mm}
\begin{center} $\Phi_s\two \DDD = \ P_s(\Xi) > 1+P_b(\Xi) = \phi_b \two \DDD $ \end{center}

\noindent
where the last equality is just a direct application of Lemma \ref{lem:plus-jeden}.

For the $(\ref{9clubsuit}) \Rightarrow (\ref{9spadesuit})$ direction suppose we are given some structure $\DDD$.
If $\DDD\not\models \good \wedge \Planet(\mars)$ then $\Phi_s \two \DDD =0$ and there is nothing to prove. So suppose $\DDD\models \good \wedge \Planet(\mars)$.
\newcommand{\dumboverleafasd}{($\ref{9clubsuit}$)\quad\qed}
\begin{alignat*}{4}
\phi_b \two \DDD
\;\;=\;\;& &&\textstyle\sum_{h\in \trip{\mmm}{}}\;  &&\nawn{\phi_b[h] \two \DDD}&&\tag{\cref{lem:pi-rachunek-sredni}}\\
\;\;\geq\;\;& \nawni{\phi_b[\bar\venus] \two \DDD} \;+\; &&\textstyle\sum_{m=1}^{\mmm}\;  &&\nawn{\phi_b[\bar\mars^m] \two \DDD}&&\tag{$\singleton{\venus,\mars}\subseteq\Planets(\DDD)$}\\
\;\;\geq\;\;& \phantom{aaaaaaaaa,}1 \;+\; &&\textstyle\sum_{m=1}^{\mmm} &&M^b_m(\Xi^{\tiny\mars})&&\\
\;\;\geq\;\;& && \textstyle\sum_{k=1}^{\kkk} &&M^s_k(\Xi^{\tiny\mars})&& \;\;=\;\; \Phi_s \two \DDD\tag*{\dumboverleafasd}
\end{alignat*}

%% file: 05-trudne.tex
\section{Proof of Theorem \ref{th:cq-cq}} \label{sec:cq-cq}


In the proofs of Theorems \ref{th:tak-samo-trudne} and \ref{th:UCQ-duzo} we never needed to bother about $\trip{}{2\leq}$. When proving the equivalence (\ref{ding80})
(in the case of Theorem \ref{th:tak-samo-trudne}) and  (\ref{9clubsuit})$\Leftrightarrow$(\ref{9spadesuit}) (for Theorem \ref{th:UCQ-duzo})  in the
``easy'' direction, we only needed to consider one structure, which was a marsification, so the set $\trip{}{2\leq}$ was empty there. When proving the ``more difficult'' implication
we (again in both cases) exploited the fact that only our b-query was a CQ-ization of a UCQ. And, for each of these two proofs, in order to show that the result of the application of  the b-query
is always big enough, we could afford  to only consider the trips from $\trip{}{1}$.

Now the situation is going to be different: also our s-query will be a  CQ-ization of a UCQ.
This means that, when proving the $(\ref{10clubsuit}) \Rightarrow  (\ref{10spadesuit}) $ implication (below) we will need to
take into account the  possibility that trips from $\trip{}{2\leq}$ contribute to the result of applying the s-query to a structure.
A new technique is introduced in this section, to deal with this issue.

\subsection{The source of undecidability and the queries $\beta_s$ and $\beta_b$}

 First of all, we are given some $\varepsilon>0$, which is fixed from now on.
We assume that $\varepsilon \leq 1$. This assumption is not crucial,
we could easily live without it, but it will save us some notations.
And anyway, the most interesting values of $\varepsilon$ are the ones just above 0, and 1 which is important for the proof of Corollary \ref{th:UCQ-dwa}.
Let $\ccc=1+\varepsilon$ and $\cccc$ be any rational number such that $\ccc > \cccc > \sqrt{\ccc}$.

As the source of undecidability we are again going to use a version of Hilbert's 10th Problem.

\begin{Fact}\label{source-of-und-dwa}
{\em The following problem is undecidable:}\smallskip

\noindent
Given are
 polynomials, $P_s$ and $P_b$,  such that the inequality 
 $\cccc \cdot  \coef(M,P_s)\leq  \coef(M,P_b) $ holds  for each monomial $M$.
\noindent
Is the inequality  $\ccc \cdot (1+P_s(\Xi)) \leq  1+ P_b(\Xi)$ satisfied for every  valuation $\Xi\;$?
\end{Fact}

For the proof of Fact \ref{source-of-und-dwa} (assuming undecidability of Hilbert's 10th Problem) see \cref{app:B}. 
From now on, till the end of this section,  we assume that 
$P_s=\Sigma_{k=1}^{\kkk}\;M^s_k$ and $P_b=\Sigma_{m=1}^{ \mmm}\;M^b_m$ are two fixed
polynomials,
  as in Fact \ref{source-of-und-dwa}. 
Now, in order to prove Theorem \ref{th:cq-cq}, we need to construct 
conjunctive queries  $\beta_s$, and  $\beta_b$ such that the two conditions are equivalent:

\noindent
\begin{minipage}{0.7\textwidth}
    \begin{equation}
        \ccc \cdot (1+P_s(\Xi)) \leq 1+P_b(\Xi)\;\; \text{holds for every  valuation.}\label{10clubsuit}
    \end{equation}
\end{minipage}
\begin{minipage}{0.299\textwidth}
    \begin{equation}
        \ccc \cdot \beta_s  \dbcontainednt \beta_b\label{10spadesuit}
    \end{equation}
\end{minipage}

\noindent
At this point, it is certainly not going to surprise the careful Reader that we are going to define:
\[ \beta_s \;=\; \good \,\wedge\, \Planet(\mars)\,\wedge\, \cqize( \gammas) \quad\quad\text{and}\quad\quad \beta_b\;=\;\eta_1 \,\wedge\, \cqize(\gammab)\]
where $\ucqize{P}$ is as in Definition \ref{def:poly-to-cq} and  $\eta_1$ is a formula very similar (but not identical) to the $\eta_0$ from Section \ref{sec:tak-samo-trudne}, namely 
$\eta_1= V(\venus, \wildcard) \wedge R(v,v) $.


\subsection{The easy direction: proof that $\neg(\ref{10clubsuit})\Rightarrow \neg(\ref{10spadesuit})$. }

Suppose  $\Xi$ is the valuation which is a counterexample to $(\ref{10clubsuit})$.
We need to produce a  $\DDD$ satisfying 
$\ccc \cdot (\beta_s \two \DDD) \leq \beta_b \two \DDD$.
It is probably obvious at this point, that this $\DDD$
is going to be $\planetize(D_\Xi)$.
Clearly, $\planetize(D_\Xi)\models \good \,\wedge\, \Planet(\mars)$. It is also easy to see 
that $\eta_1 \two \planetize(D_\Xi) =1$, so (by Observation \ref{obs:12}):

\vspace{1mm}
\noindent
We just need to show that: 
\hfill $ \ccc \cdot \nawn{ \cqize ( \gammas) \two    \planetize(D_\Xi)} > \cqize ( \gammab) \two    \planetize(D_\Xi) $

\vspace{1mm}
\noindent
Using Lemma \ref{lem:plus-jeden} this is equivalent to: \hfill
$ \ccc \cdot(1+ \nawn{     \gammas \two   D_\Xi      })> 1+ \nawn{     \gammab \two   D_\Xi      } $

\vspace{1mm}
\noindent
But, by Lemma \ref{lem:poly-reprezentacja}, this is exactly equivalent to $\neg(\ref{10spadesuit})$.

\vspace{1mm}
Proof of the $(\ref{10clubsuit}) \Rightarrow  (\ref{10spadesuit}) $ implication is much harder and  {\bf will occupy the rest of 
this section}. 


\subsection{Very good structures and what if $\DDD$ isn't one}

We now assume  $(\ref{10clubsuit})$  and we fix a structure $\DDD$. 
Our goal is to prove that $\ccc \cdot \nawni{\beta_s \two \DDD} \leq \beta_b \two \DDD$. 

Notice that 
we can assume that $\DDD$ is good and that $D\models \Planet(\mars)$, otherwise $\beta_s \two \DDD=0$ and there is nothing to prove.
Since $\DDD$ is good, $R(\venus, \venus)$ is true in $\DDD$ and $\DDD\models \eta_1$.

Let us define a structure to be 
{\em very good} if it is good, $\venus$-foggy, and such that $\venus$ is the only vertex
satisfying  $R(v,v)$. Then $\eta_1 \two \DDD = 1 $ if $\DDD$ is very good, and  $\eta_1 \two \DDD \geq  2 $
otherwise.

We of course cannot (yet) assume that $\DDD$ is very good.

\begin{lemma}\label{lem:nie-doskoczy}
$\cqize(\gammas)\two \DDD \leq \cqize(\gammab)\two \DDD$
\end{lemma}

\myproof
Since the inequality $\cccc \cdot  \coef(M,P_s)\leq  \coef(M,P_b) $ holds for each monomial $M$, we can 
imagine, w.l.o.g. that  $\cqize(\gammab) = \cqize(\gammas) \wedge \alpha$, for a CQ  $\alpha$.

Such $\alpha$ will be the  conjunction of conjunctive queries  $\cqize(\mu{M})$, for 
the (occurrences of) monomials $M$ which are in $ P_b$ but {their coefficients}
``stick out, above {those in} $P_s$'' (recall that  $\cccc \cdot  \coef(M,P_s)\leq  \coef(M,P_b) $), and of atoms 
of relation $R$
forcing the aliens in $\alpha$ to form an $R$-clique together with the aliens in  $\cqize(\gammas)$.

Now, consider a function $F:Hom(\cqize(\gammas), \DDD )\rightarrow Hom(\cqize(\gammab), \DDD )  $ defined as:
$$(H(h))(v)=h(v) \;\text{if}\; v\in var(\cqize(\gammas))\; \text{and} \;(H(h))(v)=\venus \;\text{otherwise}  $$
Such $H$ is a 1-1 mapping, which ends the proof of the lemma.
\qed

Suppose $\DDD$ is not very good. Then $\eta_1(\DDD)\geq 2$. Now recall that $\ccc\leq 2$
and use Lemma \ref{lem:nie-doskoczy} to see that in this case  $\ccc \cdot \nawn{\beta_s \two \DDD}$ indeed cannot be greater than $\beta_b \two \DDD$.

So from now on we can, and will, assume that $\DDD$ is very good. And we need to prove, under this
assumption, that:
\vspace{-2mm}
\begin{equation}
\ccc \cdot \nawn{\cqize(\gammas)\two \DDD} \;\;\leq \;\;\cqize(\gammab)\two \DDD \label{ding93}
\end{equation}

\subsection{A lemma about two sorts of trips}

We know, from Lemma \ref{lem:pi-rachunek-sredni}, that:
$$\cqize(\gammas)\two \DDD\; = 
\sum\limits_{h\in \trip{\kkk}{}} \nawni{\cqize(\gammas)[h] \two \DDD}\;\; \text{and} 
\;\;\cqize(\gammab)\two \DDD \;=
\sum\limits_{h\in \trip{\mmm}{}} \nawni{\cqize(\gammab)[h] \two \DDD}$$

\noindent
So, the implication $(\ref{10clubsuit}) \Rightarrow  (\ref{10spadesuit}) $ will be proven once we can show the following: 


\begin{lemma}\label{lem:wycieczki}
    \begin{alignat*}{5}
\ccc \;\cdot\;& \sum_{h\in \trip{\kkk}{1}\cup\{\tiny \bar\venus\}}&&\nawn{\cqize(\gammas)[h] \two \DDD} 
\quad&&\leq\quad 
\sum_{h\in \trip{\mmm}{1}\cup\{\tiny \bar\venus\}} &&\nawn{\cqize(\gammab)[h] \two \DDD}\\
%
\ccc \;\cdot\;& \;\;\sum_{h\in \trip{\kkk}{2\leq }}&&\nawn{\cqize(\gammas)[h] \two \DDD} 
\quad&&\leq\quad \;\;\,\sum_{h\in\trip{\mmm}{2\leq }} &&\nawn{\cqize(\gammab)[h] \two \DDD}\\
\end{alignat*}
\end{lemma}


\vspace{-3mm}
\noindent
Proof of \cref{lem:wycieczki}.1  is not much different than 
proof of the $(\ref{9clubsuit}) \Rightarrow (\ref{9spadesuit})$ direction in Section 
\ref{sec:drugi-dowod}: for each $p\in\Planetsbezv(\DDD)$ separately we 
exploit the assumption that 
 $\ccc \cdot (1+P_s(\Xi^p)) \leq 1+P_b(\Xi^p)$.
 See \cref{app:wycieczki} for details.

\subsection{Proof of Lemma \ref{lem:wycieczki}, the case of $\trip{}{2\leq}$}

Let $A\subseteq \Planetsbezv$, and let
$\trip{\jjj}{A}$ denote $\set{h\in \trip{\jjj}{2\leq} \mid Im(h)\cup \singleton{\venus} = A\cup \singleton{\venus}}$. 
In words, a trip $h$ is in the set $\trip{\jjj}{A}$  if  $A$ is exactly the set of destination planets (not including $\venus$) to which $h$ maps the aliens.

\vspace{1mm}
Clearly, $\trip{\mmm}{2\leq}=\bigcup_{A\subseteq \Planetsbezv; |A|\geq 2}\; \trip{\mmm}{A}  \; $ and if $A\neq A'$ then the sets $\trip{\mmm}{A} $ and $\trip{\mmm}{A'} $ are disjoint. 
And the analogous statement is true about $\trip{\kkk}{2\leq}$.
Therefore,
in order to finish the proof it will be enough to show that for each set of destination planets
$A$, such that  $ |A|\geq 2$, it holds that:
\begin{align}
\ccc  \cdot\; \;\;\sum_{h\in \trip{\kkk}{A}} \;\;\nawn{\cqize(\gammas)[h] \two \DDD} \quad\leq\quad 
\;\;\sum_{h\in \trip{\mmm}{A}}\;\; \nawn{\cqize(\gammab)[h] \two \DDD } \label{ding96}
\end{align}

Let us now fix a {set of planets} $A$ as specified above.

Recall that $\DDD$ is very good, so $\venus$ is the only planet for which 
$\DDD\models R(v,v)$ holds and, in consequence, if $h \in\trip{\mmm}{A}$ 
(or if $h \in\trip{\kkk}{A}$)
then,
for each  $p\in A$,  exactly one alien  is mapped to $p$ by $h$. This leads to:
%

\begin{definition}
Let $h \in\trip{\jjj}{A}$. We define $\bar h:A\rightarrow \{1,2,\ldots \jjj\}$ 
as:
$\;\bar h(p) = j\;$ if and only if $\;h(x_j) = p$.
\end{definition}

\noindent
We think of $\bar h$ as of $h$ seen from the perspective of the host planets.
Unlike $h$ it does not answer the question ``where am I going?''. It  answers, for each planet in $A$,
the question ``who is coming here?''.
But the the question a planet really wants to know the answer to is not ``who is coming here?'' 
but (recall that each alien travels with her own query to evaluate, and this query is 
$\mu{M}$  for some monomial $M$)  ``which monomial is coming here?''. This motivates:

\begin{definition} Recall that that $M^s_k$ is the $k$-th monomial of $P_s$  and 
 $M^b_m$ is the $m$-th monomial of $P_m$.

For $h \in\trip{\mmm}{A}$ (or $h \in\trip{\kkk}{A}$) we define $\hat h:A\rightarrow \MMM$ as:
$\hat h(p)=M^b_{\bar h(p)}$  (or, respectively, $\hat h(p)=M^s_{\bar h(p)}$).
\end{definition}

Now, the inequality (\ref{ding96})  can be equivalently rewritten as:

\begin{align*}
\ccc \cdot \sum_{\tau:A\rightarrow \MMM} \sum_{h\in \trip{\kkk}{A}, \hat h=\tau} \bignawn{\cqize(\gammas)[h] \two \DDD}
\leq
\sum_{\tau:A\rightarrow \MMM}  \sum_{h\in \trip{\mmm}{A}, \hat h=\tau} \bignawn{\cqize(\gammab)[h] \two \DDD}
\end{align*}

In order to prove this inequality, it will be of course enough, to show that for
each $\tau:A\rightarrow \MMM$:
\begin{align}
\ccc \cdot  \sum_{h\in \trip{\kkk}{A}, \hat h=\tau} \bignawn{\cqize(\gammas)[h] \two \DDD}
\leq
 \sum_{h\in \trip{\mmm}{A}, \hat h=\tau}
 \cqize(\gammab)[h] \two \DDD \label{ding94}
\end{align}

So let us now fix such $\tau$ and see what happens.
\begin{definition}
    Let: 
    $$\rrr \;=\; \prod_{p\in A} \nawn{(p \cansee \mu{\tau(p)}) \two D}\quad\quad
    \ttt_s \;=\; |\{h\in \trip{\kkk}{A}\; : \: \hat h=\tau  \}|\quad\quad
    \ttt_b \;=\; |\{h\in \trip{\mmm}{A}\; : \: \hat h=\tau  \}|.$$
\end{definition}

Now recall how $\cqize(\gammas)[h_s] \two \DDD$ is calculated 
(or $\cqize(\gammab)[h_b] \two \DDD$). For each planet $p$ of $A$ the monomial $\hat h_s(p)$
is evaluated at $p$. All the remaining monomials of $P_s$ are evaluated at $\venus$, but such evaluation 
always returns 1. Then all the results of the evaluations are multiplied. It does not matter which alien
went to which $p$ as long as they traveled with the same monomial. And it does not matter how many 
aliens stayed at home, since they all return 1 anyway. This\footnote{For full proof, see \cref{app:D}.} leads to:

\begin{lemma}\label{lem:rowne-rrr}
Take any two trips $h_s \in \trip{\kkk}{A}$ and  $h_b\in \trip{\kkk}{A}$ such that 
$\hat h_s = \hat h_b = \tau$ then:
$$\bignawn{\cqize(\gammas)[h_s] \two \DDD} \;=\; \rrr \;=\; \bignawn{\cqize(\gammab)[h_b] \two \DDD}.$$
\end{lemma}

Directly from Lemma \ref{lem:rowne-rrr} we get that:
$$ \sum_{h\in \trip{\kkk}{A},\; \hat h=\tau} \bignawn{\cqize(\gammas)[h] \two \DDD}  \quad=\quad \ttt_s\cdot \rrr \quad\;\;\text{and}\;\;\quad 
 \sum_{h\in \trip{\mmm}{A},\; \hat h=\tau} \bignawn{\cqize(\gammab)[h] \two \DDD}  \quad=\quad \ttt_b\cdot \rrr.$$
This means that, in order to prove (\ref{ding94}) we just need to show that:
\begin{lemma}\label{lem:kombinatoryczny}
    $\ccc\cdot \ttt_s \leq \ttt_b $
\end{lemma}

\noindent
How do we prove Lemma \ref{lem:kombinatoryczny} ? Here is the idea: if a planet $p$ already knows 
which monomial should be coming there, there are at least $\cccc$ times more ways of selecting an alien from 
$\cqize(\gammab)$ who owns  such monomial than an alien from 
$\cqize(\gammas)$ who owns this monomial. Since there are least two planets in $A$, there are 
least $\cccc\cdot\cccc>\ccc$ more ways. For a more detailed proof see \cref{app:E}.

%% file: aappendix-X.tex
\newpage
\section{Proof of Lemma \ref{lem:plus-jeden}}\label{app:X}

All the $\jjj$-trips in this proof are trips to $\planetize(D)$, so we write
$\trip{\jjj}{}$ instead of $ \trip{\jjj}{}(\planetize(D))$.

\noindent
We know, from Lemma \ref{lem:pi-rachunek-sredni}  that:
$\cqize(\Phi) \two \planetize(D) \;=\;  \sum_{h\in \trip{\jjj}{}}\; \nawni{\cqize(\Phi)[h] \two \planetize(D)}$.

\noindent
Notice that $\planetize(D)$ is $\venus$-foggy. 
From Lemmas \ref{lem:czarna-dziura} and  \ref{lem:pi-rachunek-maly} this implies that: 
$\cqize(\Phi)[\bar \venus] \two \planetize(D)\;=\;1$.

\noindent
Notice also that $\mars$ is the only non-$\venus$ planet of $\planetize(D)$
so $\trip{\jjj}{}=\{\bar\venus\}\cup \trip{\jjj}{1}$ and
  $\trip{\jjj}{1}=\{\bar\mars^j: 1\leq j\leq \jjj \}$ (recall that $\bar\mars^j$ is the trip  sending $x_j$ to Mars and keeping everyone else on Venus).
Thus:  
$$\textstyle\sum_{h\in \trip{\jjj}{}}\;  \nawni{\cqize(\Phi)[h] \two \planetize(D)} \;=\; 1+\sum_{h\in \trip{\jjj}{1}}\;  \nawni{\cqize(\Phi)[h] \two \planetize(D)}
\;=\;
1+\sum_{j=1}^{\jjj} \nawni{\cqize(\Phi)[\bar \mars^j] \two \planetize(D)}.$$
Now we just need to show that:\;
$ \sum_{j=1}^{\jjj} \nawni{\cqize(\Phi)[\bar \mars^j] \two \planetize(D)} \;=\; \Phi \two D$. 

\noindent
But $\Phi \two D = \Sigma_{j=1}^{\jjj}\nawni{\phi_j \two D}$, so it will be enough to prove that for each $1\leq j \leq \jjj$:
\begin{align*}
    \cqize(\Phi)[\bar \mars^j] \two \planetize(D)  \;\;=\;\; \phi_j \two D
\end{align*}
\noindent
Recall, by Lemma \ref{lem:pi-rachunek-maly}:\;\; $$\textstyle\cqize(\Phi)[\bar \mars^j] \two \planetize(D) \;=\; \nawn{\phi_j \two seen(\mars,\planetize(D))} \cdot \prod_{i\neq j} \nawn{\phi_i \two seen(\venus,\planetize(D) )}$$

\medskip
And  (by Lemma \ref{lem:czarna-dziura}):
$  \prod_{i\neq j} \;\nawn{\phi_i \two seen(\venus,\planetize(D) )}\;=\; 1$.
So what remains for us to prove is that 
$\phi_j \two D \;=\;  \phi_j \two seen(\mars,\planetize(D))  $. 
However, we know that $D= seen(\mars,\planetize(D))$. Why is it so? It clearly follows from the definition of $\planetize(D)$ that each atomic formula of $D$ is indeed true in $seen(\mars,\planetize(D))$. But to see that also the opposite 
inclusion holds, one needs to recall (from  Definition \ref{def:good}) that  formula $\good$, 
 adds, to $\planetize(D)$, some atomic facts not present in $D$. Each of these facts however contains $\venus$ as one of its arguments\footnote{Here is the reason why formula $\Atoms$ needed to be defined in such a weird way and, in consequence, why we needed the notion of pleasant queries.} And, in consequence, none of this facts  can be seen from $\mars$.
\qed


%% file: aappendix-A.tex
\newpage
\section{The $(\ref{eq:7s})\Rightarrow (\ref{eq:7h})$ part of the proof of Theorem \ref{th:tak-samo-trudne}}\label{app:A}

In order to prove this implication we need, for given 
CQ $\psi_s$ and a UCQ
$\Psi_b$ construct a pleasant 
CQ $\psi'_s$ and a pleasant UCQ
$\Psi'_b$ such that:
\begin{align}
\psi_s \dbcontained \Psi_b \quad\iff\quad \psi'_s \dbcontained \Psi'_b \label{ding44}
\end{align}
Suppose $\Sigma$ is the schema of $\psi_s$ and $\Psi_b$. Define new
schema $\Sigma'$ as follows: if $R$ is an arity $i$ relation in $\Sigma$
then there is a relation $R'$, of arity $i+1$ in $\Sigma'$.
To present an atom of $R'$, instead of $R'(a_0,a_1,\ldots a_i)$ we will write $R'(a_0)(a_1,\ldots a_i)$ and we will read it as 
``$a_0$ believes that $R(a_1,\ldots a_i)$ is true''.

Now, for a CQ $\phi=\bigwedge_{i\in I} P_i(\bar y_i)$ over $\Sigma$ define 
the query $\phi'$ over $\Sigma'$ as 
$\bigwedge_{i\in I} P'_i(x)(\bar y_i)$ where $x$ is a new variable, that is
$x\not\in var(\phi)$. Finally, for a 
UCQ $\Phi=\bigvee_{i\in I} \phi_i$ define  $\Phi'$ as 
$\bigvee_{i\in I} \phi'_i$.

In this way, we have defined our $\psi'_s $ $ \Psi'_b$. 
Clearly, they are both pleasant. Before we show that (\ref{ding44}).
holds, let us introduce two new notations.
\begin{itemize}
    \item For a structure $D$ over $\Sigma'$, and for $c\in \vertex(D)$, we define
the structure  $D_c$ over $\Sigma$ as follows:
$$D_c\models R(\bar a) \iff D\models R'(c)(\bar a)$$
    \item For a structure $D$ over $\Sigma$, and for any $c$, we define
the structure  $D^c$ over $\Sigma'$ as follows:
$$D^c\models R'(c)(\bar a) \iff D\models (\bar a)$$
\end{itemize}
Let now $\phi$ be a CQ over $\Sigma$ and let $D$ be a structure over $\Sigma'$
then it is easy to see that:
$$\nawn{\phi' \two D} \quad=\quad \Sigma_{c\in \vertex(D)}\, \nawn{\phi \two D_c}$$
From this we have that if $\Phi=\bigvee_{i\in I} \phi_i$ is a UCQ over $\Sigma$ and $D$ 
is a structure  over $\Sigma'$ then:
\begin{alignat*}{3}
   \nawn{\Phi' \two D} \quad=\quad& \sum_{\phantom{aa}i\in I\phantom{aa}} &&\nawn{\phi'_i \two D}\\
   \quad=\quad& \sum_{\phantom{aa}i\in I\phantom{aa}}  \sum_{c\in \vertex(D)} &&\nawn{\phi \two D_c}\\
   \quad=\quad& \sum_{c\in \vertex(D)} \sum_{\phantom{aa}i\in I\phantom{aa}}\; &&\nawn{\phi \two D_c}\\
   \quad=\quad& \sum_{c\in \vertex(D)} &&\nawn{\Phi \two D_c}. 
\end{alignat*}

\medskip
\noindent
(\ref{ding44}); $(\neg \Rightarrow \neg)$.\quad
Suppose $\psi_s \not\dbcontained \Psi_b$.
Then there exists $D$ such that $\nawn{\psi_s \two D} > \nawn{\Psi_b \two D}$.
Take~any~$c\in\vertex(D)$. Then  $\nawn{\psi'_s \two D^c} > \nawn{\Psi'_b \two D^c}$ and 
hence  $\psi'_s \not\dbcontained \Psi'_b$.

\medskip
\noindent
(\ref{ding44}); $(\phantom{\neg} \Rightarrow \phantom{\neg})$.\quad
Assume that  $\psi_s \dbcontained \Psi_b$.
Take any structure $D$ over $\Sigma'$. We know from the assumption that 
for each $c\in\vertex(D)$ it holds that  $\nawn{\psi_s \two D_c} \leq \nawn{\Psi_b \two D_c}$.
So, $\Sigma_{c\in \vertex(D)} \nawn{\psi_s \two D_c} \leq  \Sigma_{c\in \vertex(D)} \nawn{\Psi_b \two D_c}$. Which implies that $\nawn{\psi'_s \two D} \leq \nawn{\Psi'_b \two D}$. \qed

%% file: aappendix-Z.tex
\section{Proof of Lemma \ref{lem:niewiemco-wenus}}\label{app:Z}

 First of all recall that:
\begin{alignat*}{1}
    (\venus\cansee \psi_s) \two \DDD \;\;&=\;\;
    \psi_s\two\ seen(\venus,\DDD)\tag{Lemma \ref{lem:widzenie}}\\
    \gamma_b[\bar \venus] \two \DDD \;\;&=\;\;  \nawn{\eta_0 \two \DDD}\;\cdot\; \textstyle\prod_{m=1}^{\mmm}\; \nawn{\phi_m\two seen(\venus,\DDD)}\tag{Observation \ref{obs:12}}
\end{alignat*}

\vspace{2.5mm}
Now let us show that:\hspace{5mm}
$\psi_s\two\ seen(\venus,\DDD) \leq  \nawn{\eta_0 \two \DDD}\;\cdot\;  \textstyle\prod_{m=1}^{\mmm}\; \nawn{\phi_m\two seen(\venus,\DDD)}$

\vspace{2.5mm}
There are two cases. Either {\bf (1)} $\DDD$ is $\venus$-foggy or {\bf (2)} it is not.

\smallskip
\noindent
  If {\bf (1)}, then: \, $\psi_s\two\ seen(\venus,\DDD)=1$ (by Lemma \ref{lem:czarna-dziura}) and
$\nawni{\eta_0 \two \DDD} \cdot \prod_{m=1}^{\mmm}\; \nawni{\phi_m\two seen(\venus,\DDD)}\geq 1$.

\noindent
 If   {\bf (2)}, then  $\eta_0 \two D\geq \mmm$.
So we only need to show that
$\psi_s\two\ seen(\venus,\DDD) \leq  \mmm\cdot \prod_{m=1}^{\mmm} \nawni{\phi_m\two seen(\venus,\DDD)}$.
We also know that for each $m$ it holds that $\phi_m\two seen(\venus,\DDD)\geq 1$.
Therefore: \vspace{-2mm}

\begin{alignat*}{3}
    \psi_s\two seen(\venus,\DDD) &\;\;\leq \;\;&&\nawn{\Psi_m&&\two seen(\venus,\DDD)}\tag{Assumption $\psi_s \bagcontained \Psi_m$}\\
    &\;\;=\;\; \textstyle\sum_{m=1}^{\mmm}\;&&\nawn{ \phi_m&&\two seen(\venus,\DDD)}\tag{Definition of $\Psi_b$}\\
    &\;\;\leq\;\; \mmm \cdot \textstyle\prod_{m=1}^{\mmm}\; &&\nawn{\phi_m&&\two seen(\venus,\DDD)}\tag*{\;\qed}
\end{alignat*}

%% file: aappendix-B.tex
\newpage
\section{Proof of Fact \ref{source-of-und-dwa}}\label{app:B}

Recall we have a fixed rational $\varepsilon > 0$, that $\ccc=1+\varepsilon$, and that 
$\cccc$ is  a rational number  such that $\ccc > \cccc \geq \sqrt{\ccc}\;$. Notice that of course $\cccc/\ccc<1$.
For convenience, let us copy Fact \ref{source-of-und-dwa} here:

\begin{customfac}{\ref{source-of-und-dwa}}
{\em The following problem is undecidable:}\smallskip

\noindent
Given are
 polynomials, $P_s$ and $P_b$,  such that the inequality 
 $\cccc \cdot  \coef(M,P_s)\leq  \coef(M,P_b) $ holds  for each monomial $M$.
\noindent
Is the inequality  $\ccc \cdot (1+P_s(\Xi)) \leq  1+ P_b(\Xi)$ satisfied for every  valuation $\Xi\;$?
\end{customfac}

\noindent
We will prove Fact \ref{source-of-und-dwa} by showing a reduction from the standard variant of Hilbert's 10th Problem:

\begin{Fact}
    \label{fact:source-for-fact}
{\em The following problem is undecidable:}\\
Given are polynomials, $P^0_s$ and $P^0_b$. Does $\ P^0_s(\Xi) \leq P^0_b(\Xi)$ hold for every  valuation?
\end{Fact}


%
Let $P_s^0$ and $P_b^0$ be two polynomials from \cref{fact:source-for-fact}. Clearly, given any two natural numbers $x$ and $y$ one can find a natural number $z$ such that: $$\frac{\cccc}{\ccc} \leq \frac{x + z}{y + z}.$$
Moreover, if the above holds for $z$ then it holds for any number $z'$ greater than $z$:\\
    if $x < y$  then we always have $\frac{x + z}{y + z} < \frac{x + z'}{y + z'}$;
    if $x \geq y$ then we have $\frac{x + z'}{y + z'} > 1$ and $\frac{\cccc}{\ccc} < 1$.\\
    
\newcommand{\zz}{\uuu}
\noindent
Let $\zz$ be any natural number such that for any monomial $M$ we have $$\frac{\cccc}{\ccc} \leq \frac{\coef(M, P_b^0) + \zz}{\coef(M,P_s^0) + \zz}$$

\noindent
By $\mathcal{M}_s$ denote all the monomials that occur in $P^0_s$.  Let:

\begin{align*}
P_b^1 \;=\; P_b^0 + \sum_{M \in \mathcal{M}_s} \zz \cdot M \quad\quad\quad
P_s^1 \;=\; P_s^0 + \sum_{M \in \mathcal{M}_s} \zz \cdot M
\end{align*}

\noindent
Clearly for any valuation $\Xi$ we have
$P_s^0(\Xi) \leq P_b^0(\Xi)$ if and only if 
$P_s^1(\Xi) \leq P_b^1(\Xi)$.

\noindent
Recall that $\ccc$ is a  rational. Let  $\ccc_N,\ccc_D\in \mathbb N $ be such that 
$\ccc = \ccc_N/\ccc_D$. Define:

$$P_b^2 = P_b^1 \cdot \ccc_D \quad\quad\quad P_s^2 = P_s^2 \cdot \ccc_N$$

\noindent
We have that:
\begin{align*}
\frac{\cccc}{\ccc} \;\;&\leq\;\; \frac{\coef(M, P_b^1)}{\coef(M,P_s^1)}\\
\frac{\cccc \cdot \ccc_D}{\ccc \cdot \ccc_N} \;\;&\leq\;\; \frac{\ccc_D \cdot \coef(M, P_b^1)}{\ccc_N\cdot\coef(M,P_s^1)}\\
     \cccc \;\;&\leq\;\; \frac{\coef(M, P_b^2)}{\coef(M,P_s^2)}\\
    \cccc \cdot \coef(M, P_s^2) \;\;&\leq\;\; \coef(M, P_b^2)
\end{align*}

Again, for any valuation $\Xi$ we have
$P_s^1(\Xi) \leq P_b^1(\Xi)$ if and only if 
$\ccc \cdot P_s^2(\Xi) \leq P_b^2(\Xi)$.\smallskip

Finally note that if $\cccc \cdot \coef(M, P_s^2) \leq \coef(M, P_b^2)$ then $\cccc \cdot (\coef(M, P_s^2) -1) \leq \coef(M, P_b^2) - 1$ as $\cccc > 1$ and let $$P_b = P_b^2 - 1 \quad\quad\quad P_s^2 - 1 = P_s.$$ 
Clearly for any valuation $\Xi$ we have
$\ccc\cdot(P_s^2(\Xi)) \leq P^2_b(\Xi)$ \iffi 
$\ccc\cdot(P_s(\Xi) + 1) \leq P_b(\Xi) + 1$. \qed

%% file: aappendix-C.tex
\newpage
\section{Proof of \cref{lem:wycieczki}, for the case of $\trip{\kkk}{}\setminus\trip{\kkk}{\leq 2}$}\label{app:wycieczki}
\newcommand{\pbv}{\Planetsbezv}
\newcommand{\pbvs}{|\Planetsbezv|}
\newcommand{\sumpbv}{\sum_{p\in\pbv}}
\newcommand{\sumjj}{\sum_{j=1}^{\jjj}}
\newcommand{\sumtj}{\sum_{\phantom{am}h\in\trip{\jjj}{1}\phantom{am}}}
\newcommand{\sumbd}{\sum_{\phantom{a}h\in\trip{\jjj}{}\setminus \trip{\jjj}{2\leq}}\phantom{a}}
\newcommand{\sumbdsS}{\sum\limits_{h\in\trip{\kkk}{}\setminus \trip{\kkk}{2\leq}}}
\newcommand{\sumbdsB}{\sum\limits_{h\in\trip{\mmm}{}\setminus \trip{\mmm}{2\leq}}}
\newcommand{\sumbdsJ}{\sum\limits_{h\in\trip{\jjj}{}\setminus \trip{\jjj}{2\leq}}}

Recall that by $\pbv$ we mean $\pbv(\DDD)$ for a very good structure $\DDD$.

\begin{observation}\label{obs:app-wycieczki}
For every polynomial $P$ with $\jjj$ monomials:
    $$\pbvs - 1 \;+\; \sumbd \nawn{\cqize (\ucqize{P})[h] \two \DDD} \;\;=\;\; \sumpbv(1 + P(\Xi^p)) $$
\end{observation}

\begin{proof}

\begin{alignat*}{3}
    (-1) \;+\!\!\!\!\! \sumbd\!\!\!\! \nawn{\cqize(\ucqize{P})[h] \two \DDD} \;\;=& \sumtj &&\;\;\nawn{\cqize(\ucqize{P})[h] &&\two \DDD}\tag{1}\\
    =& \sumpbv \sumjj &&\;\;\nawn{\cqize(\ucqize{P})[\bar p^j] &&\two \DDD}\tag{Definition of $\trip{\jjj}{1}$}\\
    =& \sumpbv \sumjj &&\;\;\nawn{(p \cansee \mu{M_j}) &&\two \DDD}\tag{2}\\
    =& \sumpbv \sumjj  &&\;\;M_j(\Xi^p)\tag{3}\\
    =& \sumpbv &&\;\;P(\Xi^p)\tag{Split into monomials}\\
    =& \sumpbv &&\;\;P(\Xi^p)
\end{alignat*}

\begin{itemize}
    \item [(1)] We have that $\trip{\jjj}{1} \cup \singleton{\bar \venus} = \trip{\jjj}{} \setminus \trip{\jjj}{2 \leq}$. Note that $\cqize(\ucqize{P})[\bar \venus] \two \DDD \;=\; 1$ as $\venus$ is foggy (\cref{lem:pi-rachunek-maly,lem:czarna-dziura}).\qedhere
    \item [(2)] Note, by the definition of CQ-ization and $\mu{}$, we have that 
    $$\cqize(\ucqize{P}) \;\;=\;\; \galaxy_\jjj(x_1,x_2,\ldots x_\jjj)\wedge \bigwedge_{ j=1}^{\jjj} (x_j \cansee \mu{M_j}),$$
    where $M_j$ is the $j$-th monomial of $P$. By definition of $\bar p^j$ we have 
    $$\cqize(\ucqize{P})[\bar p^j] \;\;=\;\; \galaxy_\jjj(\underbrace{\venus,\, \venus,\,\ldots \venus}_{j-1},\,  p,\, \ldots  \venus)\wedge (p \cansee \mu{M_j}) \wedge \bigwedge_{\substack{i=1\\i\neq j}}^{\jjj} (\venus \cansee \mu{M_j}).$$ 
    As $p$ is a planet we know that $\galaxy(\venus,\,\ldots,\,p\ldots)$ holds in $\DDD$, and as $\DDD$ is $\venus$-foggy we know (\cref{lem:czarna-dziura}) that any CQ of the form $\venus \cansee \phi$ maps to $\DDD$ only via a single homomorphism - one that maps every variable of $\phi$ into $\venus$. From this, we know that:
    $$\cqize(\ucqize{P})[\bar p^j] \;\two\; \DDD \;\;=\;\; (p \cansee \mu{M_j}) \;\two\; \DDD.$$
    \item [(3)] From \cref{lem:na-mon-ok-local} we have $\;M(\Xi^p)=  (p \cansee \mu{M} ) \two \DDD$, for any monomial $M$ of $P$ and planet $p$ of $\DDD$.
\end{itemize}
\end{proof}
\newcommand{\XXS}{\mathbf{X}_{s}}
\newcommand{\XXB}{\mathbf{X}_{b}}

Finally:
\begin{align*}
    \sumpbv\ccc \cdot (1+P_s(\Xi^p)) \;\;&\leq\;\; \sumpbv(1+P_b(\Xi^p))\tag{1}\\
    \ccc \cdot \sumpbv (1+P_s(\Xi^p)) \;\;&\leq\;\; \sumpbv(1+P_b(\Xi^p))\tag{2}\\
    \ccc\cdot\left(\pbvs - 1 + \XXS\right) \;\;&\leq\;\; \pbvs - 1 + \XXB\tag{3}\\
    \ccc\cdot\left(\pbvs - 1\right) + \ccc\cdot\XXS \;\;&\leq\;\; \pbvs - 1 + \XXB\\
     \ccc\cdot\XXS \;\;&\leq\;\; \XXB\tag{4}\\
     \ccc \;\cdot\; \sumbdsS \nawn{\cqize(\gammas)[h] \two \DDD} \;\;&\leq\;\;
 \sumbdsB \nawn{\cqize(\gammab)[h] \two \DDD}\tag{5} 
\end{align*}
\begin{itemize}
    \item [(1)] By Assumption ($\clubsuit$), for each valuation $\Xi$ we have $\ccc \cdot (1+P_s(\Xi)) \;\leq\; 1+P_b(\Xi)$. The initial inequality is obtained by summing over $\pbv$.
    \item [(2)] $\ccc$ is independent from $\pbv$.
    \item [(3)] Simply apply \cref{obs:app-wycieczki}, and let for clarity:
    $$X_b \;=\; \sumbdsB \nawn{\cqize(\ucqize{P_b})[h] \two \DDD} \quad\quad\quad X_s \;=\; \sumbdsS \nawn{\cqize(\ucqize{P_s})[h] \two \DDD}$$
    
    \item [(4)] We took $\ccc \cdot (\pbvs - 1)$ from the left side and $\pbvs - 1$ from the right side. Therefore the inequality holds if $\pbvs > 0$. Note that $\mars$ is a planet in $\DDD$ and as $\DDD$ is non-trivial we have $\mars \neq \venus$. Thus $\mars \in \pbvs$. 
    \item [(5)] Definitions of $X_b$ and $X_s$ \qed
\end{itemize}

%% file: aappendix-D.tex
\newpage
\section{Proof of Lemma \ref{lem:rowne-rrr}}\label{app:D}

    We will show only the first equality as the other one is proven analogously. 

\begin{alignat*}{2}
    \cqize(\gammas)[h_s] \two \DDD \quad=&\quad \left(\galaxy_\kkk(h_s(x_1),\ldots,h_s(x_{\kkk}))\land \bigwedge_{k\in \{1,2,\ldots \kkk\}} h_s(x_i) \cansee \mu{M^s_k}\right) &&\;\two\; \DDD \tag{1}\\
    \quad=&\quad \bigwedge_{k \in \{1,2,\ldots \kkk\}} h_s(x_k) \cansee \mu{M^s_k} &&\;\two\; \DDD\tag{2}\\
    \quad=&\quad \bigwedge_{p \in A}p \cansee \mu{\tau(p)} \wedge \bigwedge_{k \in \{1,2,\ldots \kkk\} \setminus \bar h_s(A)} h_s(x_k) \cansee \mu{M^s_k} &&\;\two\; \DDD\tag{3}\\
    \quad=&\quad \bigwedge_{p \in A}p \cansee \mu{\tau(p)} \wedge \bigwedge_{k\in \{1,2,\ldots \kkk\} \setminus \bar h_s(A)} \venus \cansee \mu{M^s_k} &&\;\two\; \DDD\tag{4}\\
    \quad=&\quad \bigwedge_{p \in A}p \cansee \mu{\tau(p)}  &&\;\two\; \DDD\tag{5}\\
    \quad=&\quad \prod_{p \in A}p \cansee \mu{\tau(p)}  &&\;\two\; \DDD\tag{6}
\end{alignat*}
\begin{enumerate}
    \item From \cref{def:cqizacja,def:straszne-emy-dwa,def:poly-to-cq}.
    \item Since $h_s$ is a trip, $h_s(x_1),\ldots,h_s(x_{\kkk})$ are vertices 
    of $\DDD$, and there are no variables in the query $\galaxy_\kkk(h_s(x_1),\ldots,h_s(x_{\kkk}))$.
    And, again since $h_s$ is a trip $\DDD\models \galaxy_\kkk(h_s(x_1),\ldots,h_s(x_{\kkk}))$, so 
    $\galaxy_\kkk(h_s(x_1),\ldots,h_s(x_{\kkk})) \two \DDD=1$.

    \item Trip $h_s$ is a bijection when restricted to the preimage of $A$. The second 
    of the two large conjunctions ranges over indices of variables which are not mapped to $A$ via $h_s$. 
    \item Since $h_s \in \trip{\kkk}{A}$, we have that $h_s(x) = \venus$ if $h_s(x) \not\in A$.
    \item Since $\DDD$ is $\venus$-foggy, all the variables of $ \venus \cansee \mu{M^s_k}$ must be mapped to $\venus$. It is possible because $\DDD$ is good.
    \item From \cref{obs:12}. \qedhere
\end{enumerate} 
\qed

%% file: aappendix-E.tex
\newpage
\section{Proof of 
Lemma \ref{lem:kombinatoryczny}}\label{app:E}

    For an $M\in \mathcal M$ let $A_M$ be the set $\set{p \in A \mid \tau(p) = M}$. 
High school combinatorics tells us that:

\vspace{2mm}
    
\begin{observation}\label{obs:app:E}

    \begin{align*}
    \ttt_s \;=\; \prod_{M \in \mathrm{Im}(\tau)}\frac{\coef(M, P_s)!}{(\coef(M, P_s) - |A_M|)!} 
   \;\;\;\; \text{and} \;\;\;\;
    \ttt_b \;=\; \prod_{M \in \mathrm{Im}(\tau)}\frac{\coef(M, P_b)!}{(\coef(M, P_b) - |A_M|)!} 
    \end{align*}
\end{observation}

\vspace{5mm}

\noindent
Using the Observation we have.

\begin{alignat*}{2}
    \ttt_b 
    &= &&\prod_{M \in \mathrm{Im}(\tau)}\frac{\coef(M, P_b)!}{(\coef(M, P_b) - |A_M|)!}\tag{\cref{obs:app:E}}\\
    &=&&\prod_{M \in \mathrm{Im}(\tau)}
    \coef(M, P_b) \cdot (\coef(M, P_b) - 1) \cdot \ldots \cdot (\coef(M, P_b) - |A_M| + 1)\\
    &\geq\;&&\prod_{M \in \mathrm{Im}(\tau)}
    \cccc\;\cdot\;\coef(M, P_s) \cdot (\cccc\cdot\coef(M, P_s) - 1) \cdot \ldots \cdot (\cccc\cdot\coef(M, P_s) - |A_M| + 1)\tag{1}\\
    &\geq&&\prod_{M \in \mathrm{Im}(\tau)}
    \cccc^{|A_M|} \;\cdot\; \coef(M, P_s) \cdot (\coef(M, P_s) - 1) \cdot \ldots \cdot (\coef(M, P_s) - |A_M| + 1)\tag{2}\\
    &= \cccc^{|A|}  &\;\cdot\;&\prod_{M \in \mathrm{Im}(\tau)}\frac{\coef(M, P_s)!}{(\coef(M, P_s) - |A_M|)!}\tag{3}\\
    &\geq \cccc^{2}  &\;\cdot\;&\prod_{M \in \mathrm{Im}(\tau)}\frac{\coef(M, P_s)!}{(\coef(M, P_s) - |A_M|)!}\tag{$|A| \geq 2$}\\
    &\geq \ccc  &\;\cdot\;&\prod_{M \in \mathrm{Im}(\tau)}\frac{\coef(M, P_s)!}{(\coef(M, P_s) - |A_M|)!}\tag{4}\\
    &= \ccc  &\;\cdot\;& \ttt_s\tag{\cref{obs:app:E}}
\end{alignat*}

\vspace{4mm}
\begin{itemize}
    \item [(1)] From the assumption $\cccc \cdot  \coef(M,P_s)\leq  \coef(M,P_b) $ for each monomial $M$
    \item [(2)] Note, $\cccc \cdot a - i > \cccc \cdot (a - i)$ for $a \geq 0$ and any natural $i$. Also, $\coef(M, P_s) \geq 0$ for each monomial $M$.
    \item [(3)] The constant $\cccc$ is independent from $\tau$. Also, $\sum_{M \in \mathrm{Im}(\tau)}|A_M| = |A|$, by the definition of $A_M$.
    \item [(4)] From the assumption $\cccc \geq \sqrt{\ccc}$
\end{itemize}

\vspace{2mm}
This ends the proof of Lemma  \ref{lem:kombinatoryczny}. 